\documentclass[10pt,a4paper]{article}
\usepackage{amssymb}
\usepackage{amsmath}
\usepackage{amsthm}
\usepackage{latexsym}
\usepackage[dvips]{epsfig}
\usepackage{mathrsfs}

\usepackage{bm}
\usepackage{stmaryrd}
\usepackage{authblk}

\usepackage[dvipsnames]{xcolor}
\usepackage{tikz}
\usepackage{enumitem}

\theoremstyle{plain}
\newtheorem{proposition}{Proposition}
\newtheorem{lemma}{Lemma}
\newtheorem{theorem}{Theorem}
\newtheorem{assumption}{Assumption}

\newtheorem*{main}{Theorem}
\newtheorem{definition}{Definition}
\newtheorem{remark}{Remark}

\allowdisplaybreaks

\def\bma{{\bm a}}
\def\bmb{{\bm b}}
\def\bmc{{\bm c}}
\def\bmd{{\bm d}}
\def\bme{{\bm e}}

\def\bmg{{\bm g}}
\def\bmh{{\bm h}}
\def\bmi{{\bm i}}
\def\bmj{{\bm j}}
\def\bmk{{\bm k}}

\def\bmn{{\bm n}}

\def\bmq{{\bm q}}

\def\bmv{{\bm v}}

\def\bmx{{\bm x}}

\def\bmzero{{\bm 0}}
\def\bmone{{\bm 1}}

\def\bmA{{\bm A}}
\def\bmB{{\bm B}}

\def\bmK{{\bm K}}
\def\bmL{{\bm L}}

\def\bmN{{\bm N}}

\def\bmQ{{\bm Q}}

\def\bmalpha{{\bm \alpha}}
\def\bmbeta{{\bm \beta}}

\def\bmeta{{\bm \eta}}

\def\bmomega{{\bm \omega}}

\def\bmsigma{{\bm \sigma}}

\def\bmell{{\bm \ell}}

\def\bmpartial{{\bm \partial}}

\def\bmell{{\bm \ell}}

\DeclareMathOperator\arctanh{arctanh}
\DeclareMathOperator\sech{sech}

\usepackage{mathtools}
\usepackage{xparse}
\makeatletter
\newcommand{\raisemath}[1]{\mathpalette{\raisem@th{#1}}}
\newcommand{\raisem@th}[3]{\raisebox{#1}{$#2#3$}}
\makeatother
\NewDocumentCommand{\newrbar}{O{0pt} O{0pt}}{
  \ensuremath{\mathrlap{\raisemath{#2}{\hspace*{#1}{\mathchar'26\mkern-9mu}}}r}}

\newcounter{mnotecount}

\newcommand{\mnotex}[1]
{\protect{\stepcounter{mnotecount}}$^{\mbox{\footnotesize $\bullet$\themnotecount}}$ 
\marginpar{
\raggedright\tiny\em
$\!\!\!\!\!\!\,\bullet$\themnotecount: #1} }

\newcounter{mnote}

\usepackage{color}
\usepackage{xcolor}
\usepackage[normalem]{ulem}
\usepackage{soul}
\usepackage[hypertexnames=false]{hyperref}

\begin{document}

\title{\textbf{A comparison of Ashtekar's and Friedrich's formalisms of spatial infinity}}
 
\author[1]{ Mariem Magdy Ali Mohamed  \footnote{E-mail
    address:{\tt m.m.a.mohamed@qmul.ac.uk}}}
\author[1]{Juan A. Valiente Kroon \footnote{E-mail address:{\tt j.a.valiente-kroon@qmul.ac.uk}}}

\affil[1]{School of Mathematical Sciences, Queen Mary, University of London,
Mile End Road, London E1 4NS, United Kingdom.}

\maketitle

\begin{abstract}
Penrose's idea of asymptotic flatness provides a framework for understanding the asymptotic structure of gravitational fields of isolated systems at null infinity. However, the studies of the asymptotic behaviour of fields near spatial infinity are more challenging due to the singular nature of spatial infinity in a regular point compactification for spacetimes with non-vanishing ADM mass. Two different frameworks that address this challenge are Friedrich's cylinder at spatial infinity and Ashtekar's definition of \emph{asymptotically Minkowskian spacetimes at spatial infinity} that give rise to the 3-dimensional \emph{asymptote at spatial infinity $\mathcal{H}$}. Both frameworks address the singularity at spatial infinity although the link between the two approaches had not been investigated in the literature. This article aims to show the relation between Friedrich's cylinder and the asymptote at spatial infinity. To do so, we initially consider this relation for Minkowski spacetime. It can be shown that the solution to the conformal geodesic equations provides a conformal factor that links the cylinder and the asymptote. For general spacetimes satisfying Ashtekar's definition, the conformal factor cannot be determined explicitly. However, proof of the existence of this conformal factor is provided in this article. Additionally, the conditions satisfied by physical fields on the asymptote $\mathcal{H}$ are derived systematically using the conformal constraint equations. Finally, it is shown that a solution to the conformal geodesic equations on the asymptote can be extended to a small neighbourhood of spatial infinity by making use of the stability theorem for ordinary differential equations. This solution can be used to construct a conformal Gaussian system in a neighbourhood of $\mathcal{H}$.
\end{abstract}


\section{Introduction}
\label{Section:Introduction}
The theory of isolated systems plays a central role in astrophysical applications of Einstein's theory of General Relativity. A particularly influential approach to this theory is through Penrose's definition of \emph{asymptotic simplicity} ---see e.g. \cite{PenRin86,CFEBook}:
\begin{definition}
\label{Definition:AsymptoticSimplicity}
A vacuum spacetime $(\tilde{\mathcal{M}},\tilde{\bmg})$ is said to be asymptotically simple if there exists a smooth, oriented, time-oriented and causal spacetime $(\mathcal{M},\bmg)$ and a $C^{\infty}$ function $\Xi$ on $\mathcal{M}$ such that:
\begin{enumerate}[label=(\roman*)]
    \item $\mathcal{M}$ is a manifold with boundary $\mathscr{I} \equiv \partial \mathcal{M}$;
    \item $\Xi > 0$ on $\mathcal{M} \backslash  \mathscr{I}$ and $\Xi = 0$, $\mathbf{d}{\Xi} \neq 0$ on $\mathscr{I}$;
    \item the manifolds $\tilde{\mathcal{M}}$ and $\mathcal{M}$ are related by an embedding $\phi : \tilde{\mathcal{M}} \rightarrow \mathcal{M}$ such that $\phi(\tilde{\mathcal{M}}) = \mathcal{M} \backslash \mathscr{I}$ and 
    \begin{equation}
        \phi^{*} \bmg= \Xi^2 \tilde{\bmg}; \label{conformal-transformation}
    \end{equation}
    \item each null geodesic of $(\tilde{\mathcal{M}}, \tilde{\bmg})$ starts and ends on $\mathscr{I}$.
\end{enumerate}
\end{definition}
Identifying the interior of $\mathcal{M}$ with $\tilde{\mathcal{M}}$ one can write equation \eqref{conformal-transformation}, in a slight abuse of notation, as $\bmg= \Xi^2 \tilde{\bmg}$. Condition \emph{(iv)} is of global nature, which strictly speaking makes black hole spacetimes like the Schwarzschild solution not asymptotically simple but rather \emph{weakly asymptotically simple}.

\medskip
Roughly speaking, a spacetime is said to be asymptotically simple if it admits a conformal extension similar to that of Minkowski spacetime whereby a null hypersurface (\emph{null infinity}) $\mathscr{I}$ is attached to the spacetime. Generically, $\mathscr{I}$ consists of two disjoints components $\mathscr{I}^+$ (\emph{future null infinity}) and $\mathscr{I}^-$ (\emph{past null infinity}) representing, respectively, the endpoints and startpoints of null geodesics. As stressed by Geroch \cite{Ger76}, the central aim of the definition of asymptotically simple spacetimes (and, in fact, any approach to define isolated systems) is to identify universal structures in a wide class of spacetimes which can, in turn, be used to introduce notions of physical interest ---say, conserved quantities, radiation states. The definition tries to strike a balance between ensuring the existence of enough structures to be able to extract useful conclusions and not being too restrictive to avoid that only a handful of exact solutions satisfy it.

\medskip
A central aspect of the definition of asymptotic flatness is the assumed regularity of the conformal boundary. The classic definition of asymptotic flatness assumed smoothness. However, it is now widely accepted that this is an unnecessarily strong requirement --- see e.g. \cite{Fri18}. Although to first sight, the issue of the regularity of the conformal boundary may seem like a technicality, it is in fact, it has great physical content ---an important insight of the conformal approach is that regularity questions in the conformal point of view translate to assertions on the decay of fields. It is now known that many of the key structures supplied by the notion of asymptotic simplicity are present under more relaxed regularity assumptions ---see e.g. \cite{Fri18,ChrMacSin95}.

\medskip
The definition of asymptotical simplicity only postulates the existence of null infinity, $\mathscr{I}$. However, the conformal compactification of the Minkowski spacetime in, say, the Einstein cylinder (see e.g. \cite{CFEBook}, Section 6.2) shows that the conformal boundary of this spacetime contains a further point, \emph{spatial infinity} $i^0$, corresponding to the endpoints of spacelike geodesics. While $\mathscr{I}$ is central in the discussion of radiation properties of isolated bodies in General Relativity, $i^0$ is key in the discussion of conserved quantities ---see e.g. \cite{Ger76,AshRom92}. Also, as evidenced by the stability theorems of Minkowski spacetime, the properties of the gravitational field in a neighbourhood of spatial infinity is closely related to the regularity of null infinity.

\medskip
While in Minkowski spacetime $i^0$ is a regular point of the conformal structure, it is well known that for spacetimes with non-vanishing (ADM) mass, one has a singularity of the conformal structure ---see e.g. \cite{CFEBook}, Chapter 20. This fact makes the analysis of the gravitational field particularly
challenging. To address this difficulty, it is necessary to introduce a different representation of spatial infinity which, in turn, allows to suitably resolve the structure of the gravitational field in this region. In particular, to avoid having to deal with directional dependent limits at spatial infinity, it is natural to blow up the point $i^0$ to a 2-sphere. 

\medskip
\noindent
\textbf{The hyperboloid at spatial infinity.} In an attempt to overcome the difficulties posed by the fact that the classic definition of asymptotic simplicity (see Definition \ref{Definition:AsymptoticSimplicity} above) makes no reference to the behaviour of the gravitational field at spatial infinity, in \cite{AshHan78} Ashtekar \& Hansen introduced a new definition of asymptotic simplicity in both null and spacelike directions ---see also \cite{Ash80}. This definition (\emph{asymptotically empty and flat at null infinity and spatial infinity, AEFANSI}) combines the conditions on the conformal extension given in the definition of asymptotical flatness with further assumptions on the conformal factor at spatial infinity. Essentially, it is further required that spatial infinity admits a so-called \emph{point compactification} ---see, e.g. \cite{CFEBook}, Section~11.6 for a discussion of this notion. The strategy behind the approach put forward in \cite{AshHan78} is to make use of spacetime notions rather than, say, making use of the initial value problem. Their definition of asymptotic flatness allows them to blow up $i^0$ to the timelike unit 3-hyperboloid ---the \emph{hyperboloid at spatial infinity}. The sections of the hyperboloid give the blowing up of $i^0$ to 2-spheres while time direction along the (timelike) generators of the hyperboloid can be, roughly speaking, associated with all the possible ways in which the asymptotic region of a Cauchy hypersurface can be boosted. The definition of AEFANSI spacetimes is geared towards the discussion of asymptotic symmetries and conserved quantities at spatial infinity. 

In a slightly different context, in \cite{BeiSch82,Bei84} the hyperboloid at spatial infinity is used, in conjunction with the vacuum Einstein field equations to obtain asymptotic expansions of the gravitational field near spatial infinity in negative powers of a radial coordinate. The main observation is that the field equations give rise to a hierarchy of linear equations for the coefficients of the expansion. It turns out that this hierarchy can be solved, provided the initial data satisfies certain conditions. This work can be regarded as the first serious attempt to investigate the consistency between the Einstein field equations and geometric notions of asymptotic flatness.  

The notion of the hyperboloid at spatial infinity has been revisited in \cite{AshRom92} where the definition of \emph{asymptotically Minkowskian spacetimes at spatial infinity (AMSI)} has been introduced ---see Definition \ref{Definition:AMSI} in Section \ref{Section:AMSI} of the main body of the present article. As in the analysis of~\cite{AshHan78}, the approach in \cite{AshRom92} makes use of spacetime structures and, thus, it is not geared towards the discussion of initial value problems. However, in contrast to the definition of AEFANSI spacetimes in \cite{AshHan78}, the  definition of AMSI spacetimes focuses entirely on spatial infinity and, thus, it makes no assumptions about the properties of null infinity. A consequence of the concept of AMSI spacetimes is the existence of a 3-dimensional manifold $\mathcal{H}$, the \emph{asymptote at spatial infinity}, which generalises the notion of the hyperboloid at spatial infinity. 

\medskip
\noindent
\textbf{Friedrich's cylinder at spatial infinity.} In \cite{Fri98a}, Friedrich puts forward an alternative conformal representation of the region of spacetime in the neighbourhood of spatial infinity ---\emph{the F-gauge}. The aim of this representation is the formulation of a \emph{regular initial value problem at spatial infinity} for the \emph{conformal Einstein field equations}. This initial value problem is key in the programme to analyse the genericity of asymptotically simple spacetimes and relies heavily on the properties of certain conformal invariants (\emph{conformal geodesics}), and it is such that both the equations and the initial data are regular at the conformal boundary. A central structure in this representation of spatial infinity is the \emph{cylinder at spatial infinity} which, in broad terms, corresponds to the blow-up of the traditional point at spatial infinity to a 2-sphere plus a time dimension ---hence, the cylinder at spatial infinity. In contrast to the hyperboloid at spatial infinity, the cylinder at spatial infinity has a finite extension in the time direction. The endpoints of the cylinder, the \emph{critical sets}, correspond to the points where spatial infinity meets null
infinity. Thus, the F-gauge is ideally suited to the analysis of the connection between the behaviour of the gravitational field at spatial infinity and radiative properties. This idea has been elaborated
in \cite{FriKan00} to express the Newman-Penrose constants in terms of the initial data for the Einstein field equation and to study the behaviour of the Bondi mass as one approaches spatial infinity ---see \cite{Val03b}; also \cite{Val07a}. The key property of the cylinder at spatial infinity, which allows connecting properties of the Cauchy initial data in a neighbourhood of spatial infinity with the behaviour near null infinity, is that the cylinder at spatial infinity is a total characteristic of the conformal Einstein field equations ---that is, the associated evolution equations reduce in its entirety to a system of transport equations on this hypersurface. This property precludes the possibility of prescribing boundary data on the cylinder but allows computing a particular type of asymptotic expansions which makes it possible to understand the role of certain pieces of the initial data have on the regularity properties of the gravitational at null infinity ---the so-called \emph{peeling behaviour}; see \cite{Fri98a,FriKan00,Val04a,Val04d,Val04e,Val05a,GasVal17b}.

\medskip
The F-gauge used in Friedrich's representation of spatial infinity is a gauge prescription based on the properties of conformal geodesics ---i.e. a pair consisting of a curve and a covector along the curve satisfying conformally invariant properties. A remarkable property of conformal geodesics is that in Einstein spaces, they give rise to a canonical factor that has a quadratic dependence on the parameter of the curve ---see \cite{FriSch87,CFEBook}. The resulting expression depends on certain initial data, which can be chosen in such a way that the curves reach the conformal boundary and that no caustics are formed. In this way, one can obtain a conformal Gaussian gauge in which coordinates and a frame defined on an initial hypersurface are propagated along the conformal geodesics. In effect, this procedure provides a \emph{canonical way} of obtaining conformal extensions of Einstein spaces. In the context of the initial value problem for the conformal Einstein equations, one has a gauge in which the location of the conformal boundary is known \emph{a priori}. 

\medskip
\noindent
\textbf{Melrose-type compactifications.} In the context of this introduction, it is worth mentioning the proof of the stability of Minkowski spacetime by Hintz \& Vasy in which the existence of spacetimes with \emph{polyhomogeneous asymptotic expansions} is established ---see \cite{HinVas17}--- makes use of of a compactification of spacetime into a manifold with boundary and corners. This procedure is inspired by the methods of Melrose's geometric scattering theory programme ---see e.g. \cite{Mel95}. This construction has a very strong resemblance to that introduced by Friedrich in \cite{Fri98a}. The precise relation between these two seemingly connected representations of spatial infinity will be elaborated elsewhere. 

\subsubsection*{Main results of this article}
A cursory glance at Ashtekar's and Friedrich's constructions of spatial infinity suggests that they should be closely related. That this is the case has been part of the longstanding folklore of mathematical relativity. It is the purpose of this article to establish, in a rigorous way, the connection between these two representations of spatial infinity.

\medskip
In order to carry out the objective outlined in the previous paragraph, we first analyse this relation for Minkowski spacetime, where everything can be explicitly computed. The outcome of this analysis is that Ashtekar's and Friedrich's representations are related to each other through a conformal transformation which preserves the asymptotic behaviour of the metric along spatial infinity but compactifies the time direction. In particular, the hyperboloid at spatial infinity is compactified into the $1+2$ Einstein Universe in the same way as the de Sitter spacetime is compactified into the $1+3$ Einstein Universe.  The transformation between the two representations can be expressed in terms of properties of conformal geodesics ---although in the case of Minkowski spacetime, because of its simplicity, this does not play an essential role.

\medskip
The intuition gained in the analysis of Minkowski spacetime is, in turn, used to establish the relation between Ashtekar's notion of an asymptote at spatial infinity given by the definition of an AMSI spacetime ---see Definition \ref{Definition:AMSI} in the main text. As the definition of an AMSI spacetime does not imply the existence of null infinity, to carry out our analysis, an extra assumption is required ---see Assumption \ref{Assumption:Extra} in the main text. In essence,  we assume that a sufficiently regular null infinity can be attached to the neighbourhood of spacetime around Ashtekar's asymptote. Our main result, whose proof is based on a stability argument for the solutions to the conformal geodesic equations in a neighbourhood of the asymptote is the following:

\begin{main}
Given an AMSI spacetime that can be conformally extended to null infinity, there exists a sufficiently small neighbourhood of the asymptote at spatial infinity in which is possible to construct  conformal Gaussian coordinates based on curves which extend beyond null infinity. 
\end{main}

Associated with the conformal Gaussian system whose existence is ensured by the above statement, there exists a conformal factor that provides the precise relation between Ashtekar's representation of spatial infinity and F-gauge representation. 

\subsection*{Outline of the article}
This article is structured as follows. In Section \ref{Section:ConformalMethods} we provide a brief summary of the conformal methods that are required in the analysis of this article. These include conformal geodesics, conformal Gaussian systems, the conformal Einstein field equations and their constraints. For a full account of the associated literature, the reader is referred to the monograph \cite{CFEBook}. Section \ref{Section:MinkowskiSpacetime} provides a discussion of the various representations of spatial infinity on Minkowski spacetime. In particular, it contains the explicit connection between the F-gauge and Ashtekar's representation of spatial infinity. Furthermore, it also contains a discussion on how this connection can be expressed in terms of conformal geodesics. Section \ref{Section:AMSI} reviews Ashtekar's definition of asymptotically Minkowskian spacetimes at spatial infinity (AMSI) and the associated notion of an asymptote. It further contains a discussion of the consequences of this definition in light of the conformal Einstein field equations and the associated constraint equations on timelike hypersurfaces. Subsection \ref{Section:HyperboloidToCylinder} provides a detailed discussion of the relation between Ashtekar's asymptote and Friedrich's cylinder viewed as intrinsic 3-manifolds. Section \ref{Section:CGAsymptote} contains the main analysis in the article ---namely, the construction of a conformal Gaussian gauge system in a neighbourhood of an asymptote. It also provides a more detailed statement of the main theorem of this article ---Theorem \ref{Theorem:MainPreciseFormulation}.  In addition, the article contains several appendices. Appendix \ref{Appendix:ODETheory} provides the stability theorem for solutions to a system of ordinary differential equations and their existence times which is used in the
analysis of the solutions to the conformal geodesic equations in a neighbourhood of the asymptote. Appendix \ref{Appendix:RegularityAsymptoteDetails} provides details of the proof of a technical lemma (Lemma \ref{Lemma:RegularityAsymptoteDetails}) on the regularity of geometric fields at the asymptote. Finally, Appendix \ref{Appendix:CGMinkowski} provides details on a certain class of conformal geodesics on Minkowski spacetime. 

\subsection*{Notations and conventions}
In what follows $a,\,b,\,c\ldots$ will denote spacetime abstract tensorial indices, while $i,\,j,\,k,\ldots $ are spatial tensorial indices ranging from 1 to 3. By contrast, $\bma,\, \bmb,\,\bmc,\ldots$ and  $\bmi,\, \bmj,\bmk,\ldots$ will correspond, respectively, to spacetime and spatial coordinate indices. 

\smallskip
The signature convention for spacetime metrics is $(+,-,-,-)$. Thus, the induced metrics on spacelike hypersurfaces are negative definite. The analysis of this article involves several conformally related metrics. To differentiate between them we adhere to the following conventions: the \emph{tilde} ($\tilde{\phantom{g}}$) is used to denote metrics satisfying the physical Einstein field equations; the \emph{overline} ($\bar{\phantom{X}}$) is used to denote metrics in the F-gauge; metrics \emph{without an adornment} (e.g. $\bmg$) are in the conformal gauge given by the Definition \ref{Definition:AMSI} of a spacetime asymptotically Minkowskian at spatial infinity (AMSI); metrics in a gauge with compact asymptote are denoted by a \emph{grave accent} ($\grave{\phantom{X}}$). Finally, metrics with a \emph{hat} ($\hat{\phantom{X}}$) correspond to conformal representations in which spatial infinity corresponds to a point.

\smallskip 
An index-free notation will be often used. Given a 1-form ${\bmomega}$ and a vector ${\bmv}$, we denote the action of ${\bmomega}$ on ${\bmv}$ by $\langle {\bmomega},{\bmv}\rangle$. Furthermore, ${\bmomega}^\sharp$ and ${\bmv}^\flat$ denote, respectively, the contravariant version of ${\bmomega}$ and the covariant version of ${\bmv}$ (raising and lowering of indices) with respect to a given Lorentzian metric.  This notation can be extended to tensors of higher rank (raising and lowering of all the tensorial indices). The  conventions for the curvature tensors will be fixed by the relation
\[
(\nabla_a \nabla_b -\nabla_b \nabla_a) v^c = R^c{}_{dab} v^d.
\]
Also, one can write the decomposition of the curvature tensor as 
\begin{eqnarray}
&& R^c{}_{dab} = C^c{}_{dab} + 2 \left( \delta^c{}_{[a} L_{b]d} - g_{d[a} L_{b]}{}^c \right), \label{curvature-decomposition}
\end{eqnarray}
where $C^c{}_{dab}$ is the \emph{Weyl} tensor and $L_{bd}$ is the \emph{Schouten} tensor of the metric $g_{ab}$.


\section{Tools of conformal methods}
\label{Section:ConformalMethods}
The purpose of this section is to introduce the tools of conformal methods that will be required in this article. 

\subsubsection*{General conventions and notation}
In the following assume that $(\tilde{\mathcal{M}}, \tilde{\bmg})$ denotes a spacetime satisfying the vacuum Einstein field equation
\begin{eqnarray}
\tilde{R}_{ab}=0.
\label{Einstein-Field-Equation}
\end{eqnarray}

Following standard usage, we call the pair  $(\tilde{\mathcal{M}}, \tilde{\bmg})$ the \emph{physical spacetime}, while any conformally related spacetime $(\mathcal{M}, \bmg)$ with
\[
\bmg \equiv \Xi^2 \tilde{\bmg},
\]
will be referred to as the \emph{unphysical spacetime}.  

\subsection{Conformal geodesics}
\label{Section:CG}

Given an interval $I \subseteq \mathbb{R} $ and $\tau \in I$, the curve ${x}(\tau)$ is said to be a \emph{conformal geodesic} if there exists a 1-form $\bmbeta(\tau)$ along ${x}(\tau)$ such that 
\begin{subequations}
\begin{eqnarray}
&& \tilde{\nabla}_{\dot{\bm x}} \dot{\bmx} = -2 \langle {\bmbeta},
\dot{\bmx} \rangle
\dot{\bmx} + \tilde{\bmg}(\dot{\bmx},\dot{\bmx}) {\bmbeta}^\sharp, \label{ConformalGeodesic1} \\
&& \tilde{\nabla}_{\dot{\bm x}} {\bmbeta} = \langle {\bmbeta}, \dot{\bm x}
\rangle {\bmbeta} - \tfrac{1}{2} \tilde{\bmg}^\sharp ({\bmbeta},{\bmbeta})
\dot{\bmx}^\flat + \tilde{\bm L}(\dot{\bmx}, \cdot), \label{ConformalGeodesic2}
\end{eqnarray}
\end{subequations}
are satisfied. For spacetimes $(\tilde{\mathcal{M}},\tilde{\bmg})$
satisfying the vacuum Einstein equations \eqref{Einstein-Field-Equation}, one can explicitly determine
a canonical conformal factor given initial data on an initial
hypersurface ---see \cite{CFEBook,Fri95}, Proposition 5.1. Specifically, 

\begin{proposition}
\label{canconical-confromalFactor}
Given an Einstein spacetime $(\mathcal{\tilde{M}}, \tilde{\bmg})$, a solution $(x(\tau), \bmbeta(\tau))$ to the conformal geodesic equations \eqref{ConformalGeodesic1}-\eqref{ConformalGeodesic2} and an unphysical spacetime $\bmg = \Theta^2 \tilde{\bmg}$ defined such that $\bmg(\dot{\bmx},\dot{\bmx})=1$. Then the conformal factor $\Theta$ can be written as a quadratic polynomial in terms of $\tau$, i.e.
\begin{eqnarray}
&& \Theta(\tau) = \Theta_{\star} + \dot{\Theta}_{\star}(\tau - \tau_{\star}) + \frac{1}{2} \ddot{\Theta}_{\star} (\tau - \tau_{\star})^2,
\label{canconical-conformalFactor1}
\end{eqnarray}
with 
\begin{eqnarray}
\dot{\Theta}_{\star} = \langle \bmbeta_{\star},\dot{\bmx}_{\star} \rangle \Theta_{\star}, \; \; \; \; \Theta_{\star} \ddot{\Theta}_{\star} = 2 \tilde{\bmg}^{\sharp}(\bmbeta_{\star},\bmbeta_{\star}) + \frac{1}{6} \lambda,
\label{canconical-conformalFactor2}
\end{eqnarray}
where $\lambda$ is the \emph{cosmological constant}.
\end{proposition}

\begin{remark}
    In subsequent discussions, the cosmological constant $\lambda$ is taken to be zero.
\end{remark}

\subsubsection*{Conformal Gaussian systems}
Conformal geodesics can be used to construct the so-called \emph{
Conformal Gaussian Systems} in which coordinates and adapted frames
are propagated off an initial hypersurface $\mathcal{S}$. One
constructs a conformal Gaussian system by initially introducing a
$\bmg$-orthonormal Weyl-propagated frame $\{ \bme_{\bma} \}$ along a
congruence of conformal geodesics $(x(\tau),\bmbeta(\tau))$
on $(\tilde{\mathcal{M}},\tilde{\bmg}= \Xi^{-2} \bmg )$ and choosing
the time coordinate such that $\bme_{\bm0} = \partial_{\tau}$. Then if
$(x^i)$ denotes the coordinates of a point $p$ on $\mathcal{S}$, one
can propagate the spatial coordinates off $\mathcal{S}$ by requiring
them to be constant along the conformal geodesic intersecting
$\mathcal{S}$ at $p$. Then the conformal Gaussian system is given by
$(\tau, x^i)$.

\subsection{Conformal field equations}
The vacuum \emph{conformal Einstein field equations} on the unphysical spacetime are given by
\begin{subequations}
\begin{eqnarray}
&& \nabla_a \nabla_b \Xi= - \Xi L_{ab} + s g_{ab}, \label{conformal-field-equations1}\\ 
&& \nabla_a s = - L_{ac} \nabla^c \Xi, \label{conformal-field-equations2}\\
&& \nabla_c L_{db} - \nabla_d L_{cb} = \nabla_a \Xi d^a{}_{bcd}, \label{conformal-field-equation3}\\
&& \nabla_a d^a{}_{bcd} =0, \label{conformal-field-equations4}\\
&& 6 \Xi s - 3 \nabla_c \Xi \nabla^c \Xi = 0, \label{conformal-field-equations5}
\end{eqnarray}
\label{Conformal-field-equations}
\end{subequations}
where $d^a{}_{bcd}=\Xi^{-1}C^a{}_{bcd}$ is the \emph{rescaled
Weyl} tensor and $s$ is the \emph{Friedrich's scalar} given by 
$$
s \equiv \frac{1}{4} \nabla^c\nabla_c\Xi + \frac{1}{24} R \Xi.
$$
The conformal Einstein field equations constitute a set of
differential equations for the fields $\Xi$, $s$, $L_{ab}$ and
$d^a{}_{bcd}$. A solution to the vacuum conformal field equations
\eqref{conformal-field-equations1}-\eqref{conformal-field-equations5}
implies, whenever $\Xi\neq 0$, a solution to the vacuum Einstein field
equations (\ref{Einstein-Field-Equation}) on the physical spacetime.

\subsubsection*{Conformal constraint equations}
For later use, we give here the constraint equations implied by the
conformal Einstein field equations \eqref{conformal-field-equations1}-\eqref{conformal-field-equations5} on a timelike hypersurface $\mathcal{T}$. These equations can be obtained through a projection formalism and details of the derivation can be found in \cite{CFEBook}.

Let $n^{a}$ denote the unit-normal to the timelike hypersurface $\mathcal{T}$. Then, the projection tensor is given by 
\begin{equation*}
    q_{a}{}^{b} = \delta_{a}{}^{b} + n_{a} n^{b}.
\end{equation*}
Then introduce the shorthand notation 
\begin{eqnarray}
&& \Sigma \equiv n^a \nabla_a \Xi. \label{Sigma-definition}
\end{eqnarray}
Given the above, the \emph{conformal Einstein constraint equations} in vacuum can be written in terms of $q_{ab}$, the extrinsic curvature $K_{ab}$, the Schouten tensor $L_{ab}$ and $\omega$, denoting the restriction of the conformal factor $\Xi$ to $\mathcal{T}$
\begin{subequations}
    \begin{align}
        & D_{a} D_{b} \omega + \Sigma K_{ab} + \omega (L_{ab} + 2 n_{(a} L_{b)\perp} + n_{a} n_{b} L_{\perp \perp}) - s q_{ab} =0, \label{ConformalConstraint1} \\
        & D_{a} \Sigma + \omega (L_{a \perp} + n_{a} L_{\perp \perp}) =0, \label{ConformalConstraint2} \\
        & D_{a} s - \Sigma (L_{a \perp} + n_{a} L_{\perp \perp}) =0, \label{ConformalConstraint3} \\
        & D_{a} L_{bc} - D_{b} L_{ac} + \Sigma (d_{ab\perp c} + 2 n_{[a} d_{b] \perp c \perp}) =0, \label{ConformalConstraint4} \\
        & D_{a} L_{b \perp} - D_{b} L_{a \perp} + K_{b}{}^{e} L_{ae} - K_{a}{}^{e} L_{be} + n_{a} K_{b}{}^{e} L_{e \perp} - n_{b} K_{a}{}^{e} L_{e \perp} =0, \label{ConformalConstraint5} \\
        & D^{c} d_{c \perp ab}=0, \label{ConformalConstraint6} \\
        & D^{b} d_{b \perp a \perp } + K^{ed} d_{e \perp ad} - n_{a} K^{ed} d_{e \perp d \perp} =0, \label{ConformalConstraint7} \\
        & 6 \omega s + 3 \Sigma^2=0. \label{ConformalConstraint8}
    \end{align}
    \label{Constraint-equations-asymptote}
\end{subequations}
In the above, the conformal factor $\omega$ is assumed to be constant on $\mathcal{T}$, i.e., $D_{a}\omega = q_{a}{}^{b} \nabla_{b} \omega =0$. Moreover, the symbol ${}_{\bot}$ indicates a contraction with the unit normal $n^a$, so the terms $d_{a \bot b \bot}$ and $d_{a \bot bc}$ correspond, essentially, to the electric and magnetic parts of the rescaled Weyl tensor $d_{abcd}$ with respect to the normal $n^a$ ---respectively. 

\begin{remark}
\emph{
    \begin{enumerate}
        \item The constraint equations \eqref{Constraint-equations-asymptote} can be obtained from the conformal field equations \eqref{Conformal-field-equations} by various contractions with the projection tensor $q_{a}{}^{b}$ and the unit-normal $n^{a}$. For example, equations \eqref{ConformalConstraint1} and \eqref{ConformalConstraint2} are obtained by contractions of equation \eqref{conformal-field-equations1} with $q_{c}{}^{a} q_{d}{}^{b}$ and $q_{c}{}^{a} n^{b}$, respectively.
        \item All tensors intrinsic to $\mathcal{T}$, e.g., $\bmq$ and $\bmK$, will have vanishing components in the normal direction. 
    \end{enumerate}}
\end{remark}


\section{Representations of spatial infinity in Minkowski spacetime}
\label{Section:MinkowskiSpacetime}
In this section, we review several representations of spatial infinity for Minkowski spacetime. This analysis will provide insight and motivate the analysis in curved spacetimes.

\medskip
In the following let $(\mathbb{R}^4,\tilde{\bmeta})$ denote Minkowski spacetime. Let $(\overline{x})=(x^\mu)$ denote the standard Cartesian coordinates and write $x^0=\tilde{t}$, etc. We will also make use of spherical coordinates $(\tilde{t},\tilde{\rho},\theta^{\mathcal{A}})$ where
$(\theta^{\mathcal{A}})$ denotes some choice of spherical coordinates over $\mathbb{S}^2$. In terms of the above, one has that
\begin{eqnarray*}
&& \tilde{\bmeta} = \eta_{\mu\nu} \mathbf{d}x^\mu \otimes
   \mathbf{d}x^\nu, \\
&& \phantom{\tilde{\bmeta}} =
   \mathbf{d}\tilde{t}\otimes\mathbf{d}\tilde{t} -
   \mathbf{d}\tilde{x}\otimes\mathbf{d}\tilde{x}-
   \mathbf{d}\tilde{y}\otimes\mathbf{d}\tilde{y}-
   \mathbf{d}\tilde{z}\otimes\mathbf{d}\tilde{z}, \\
&& \phantom{\tilde{\bmeta}} =
   \mathbf{d}\tilde{t}\otimes\mathbf{d}\tilde{t} - \mathbf{d}\tilde{\rho}\otimes\mathbf{d}\tilde{\rho}-\tilde{\rho}^2 \bmsigma,
\end{eqnarray*}
where $\eta_{\mu\nu} \equiv \mbox{diag}(1,-1,-1,-1)$, and $\bmsigma$
is the standard round metric over $\mathbb{S}^2$.

\subsection{Spatial infinity as point}
\label{Section:MinkowskiPointSpatialInfinity}
In first instance, we consider the standard representation of
spatial infinity as a point. Intuitively, the region in Minkowski spacetime
associated with spatial infinity is contained in the domain
\[
\tilde{\mathcal{D}} \equiv \big\{ p\in \mathbb{R}^4 \;|\;
\eta_{\mu\nu} x^\mu(p)x^\nu(p)<0
\big\},
\]
---the \emph{complement of the light cone through the origin}, see
Figure \ref{Fig:Domain-D}. Now, introducing the \emph{inversion coordinates} $\overline{y}=(y^\mu)$ defined by 
\[
y^\mu =- \frac{x^\mu}{X^2}, \qquad X^2\equiv \eta_{\mu\nu}x^\mu x^\nu,
\] 
it follows that 
\[
\eta_{\mu\nu} \mathbf{d} y^\mu \otimes \mathbf{d} y^\nu = X^{-4}
\eta_{\mu\nu} \mathbf{d} x^\mu\otimes \mathbf{d} x^\nu.
\]
The latter suggests introducing the conformal factor $\Xi\equiv 1/X^2$
so that 
\[
\hat{\bmeta} \equiv \Xi^2 \tilde{\bmeta} = \Xi^2 \mathbf{d}x^\mu\otimes\mathbf{d}x^\nu.
\]
The conformal boundary defined by $\Xi=0$ decomposes into the sets
\begin{eqnarray*}
&&\mathscr{I}^+ \equiv \big\{ p\in\mathbb{R}^4 | y^0(p)>0, \;
   \eta_{\mu\nu} y^\mu(p) y^\nu(p) =0  \big\}, \\
&& \mathscr{I}^- \equiv \big\{ p\in\mathbb{R}^4 | y^0(p)<0, \;
   \eta_{\mu\nu} y^\mu(p) y^\nu(p) =0  \big\}, \\
&&i^0 \equiv \big\{ p\in\mathbb{R}^4 |  (y^\mu(p)) =(0,0,0,0) \}.
\end{eqnarray*}
The sets $\mathscr{I}^+$ ($\mathscr{I}^-$),  can be shown to be the
endpoints of future (past) null geodesics while spatial geodesics end
up in the point $i^0$ ---\emph{spatial infinity}, located in this
representation at the origin. Observe that $\mathscr{I}^+$ and $\mathscr{I}^-$ do not contain the whole of null infinity, only
the part of the conformal boundary close to spatial infinity ---this
is a peculiarity of this conformal representation.  

\begin{remark}
In literature, $\mathscr{I}^+$ and $\mathscr{I}^-$ are usually used to denote the whole of future and past null infinity, respectively. In this setting, in a slight abuse of notation, we use $\mathscr{I}^+$ and $\mathscr{I}^-$ to denote the parts of the conformal boundary close to spatial infinity. 
\end{remark}

It can be verified that
while $\mathbf{d}\Xi|_{\mathscr{I}^\pm}\neq 0$. For $i^0$, it holds
that
\begin{equation}
\Xi(i^0) =0, \qquad \mathbf{d}\Xi(i^0) =0, \qquad \mbox{Hess}\,
\Xi(i^0)\neq 0.
\label{PointCompactificationConditions}
\end{equation}
Introduce spherical coordinates $(t, \rho,\theta^\mathcal{A})$, with
$\rho^2\equiv (y^1)^2 + (y^2)^2 +(y^3)^2$, then the
unphysical spacetime metric can be written as
\begin{eqnarray}
&& \hat{\bmeta} = \mathbf{d}t \otimes \mathbf{d} t -\mathbf{d}\rho\otimes
\mathbf{d}\rho -\rho^2 \bmsigma, \qquad \Xi =t^2-\rho^2.  \label{pointCompactified-metric}
\end{eqnarray}
For future use, it is noticed that
\[
\tilde{t}= - \frac{t}{t^2-\rho^2}, \qquad \tilde{\rho} = - \frac{\rho}{t^2-\rho^2}.
\]

\begin{figure}[t]
\centering
\includegraphics[width=95mm]{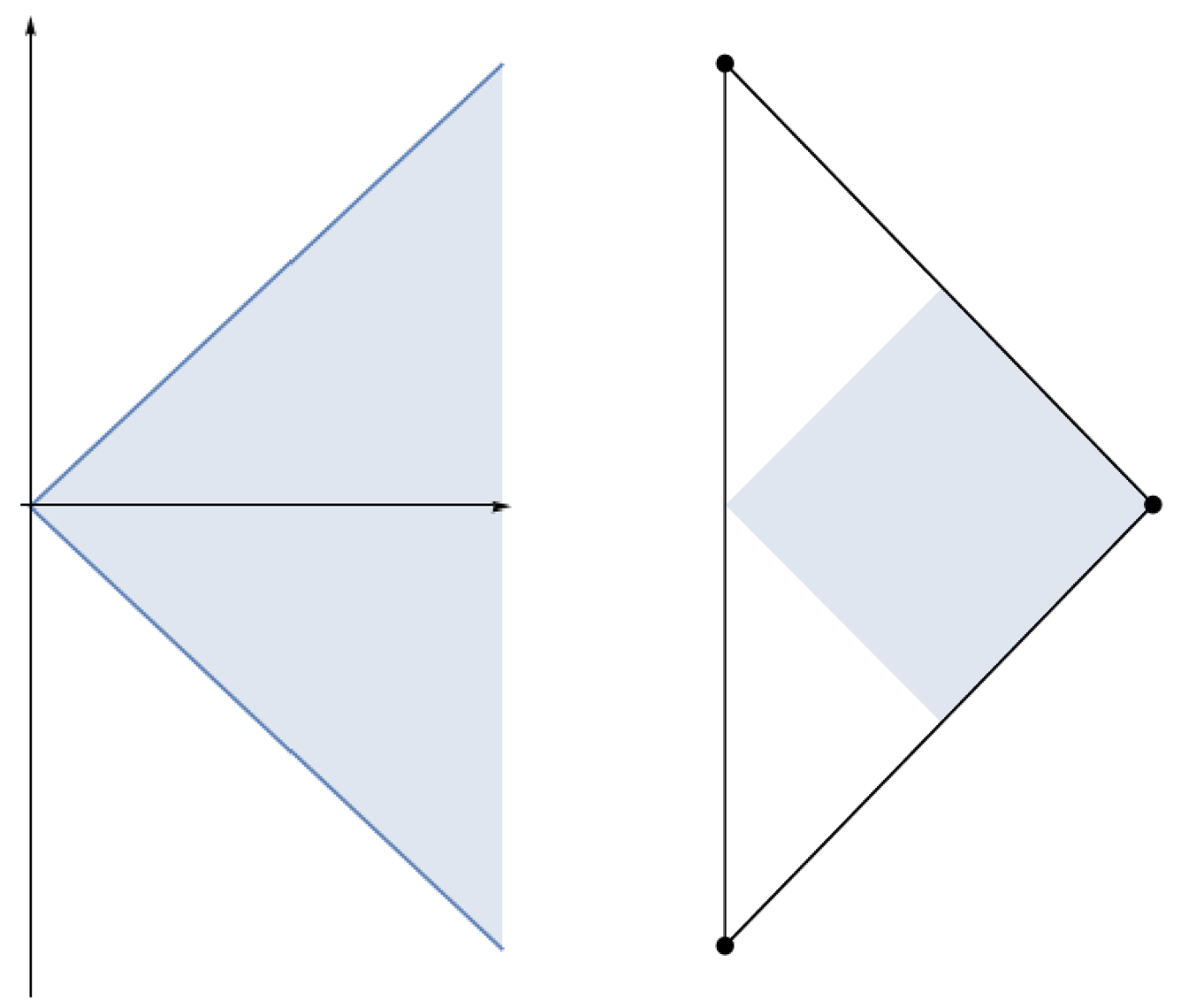}
\put(-265,233){$\tilde{t}$}
\put(-150,110){$\tilde{\rho}$}
\put(-210,90){$\tilde{\mathcal{D}}$}
\put(-60,110){$\tilde{\mathcal{D}}$}
\put(-5,110){$i^0$}
\put(-108,223){$i^+$}
\put(-108,1){$i^-$}
\put(-35,70){$\mathscr{I}^-$}
\put(-35,150){$\mathscr{I}^+$}
\put(-70,-10){(b)}
\put(-210,-10){(a)}
\caption{ (a) Diagrammatic depiction of the domain $\tilde{\mathcal{D}}$ containing spatial infinity, (b) The domain $\tilde{\mathcal{D}}$ on the conformal diagram of Minkowski spacetime.}
\label{Fig:Domain-D}
\end{figure}
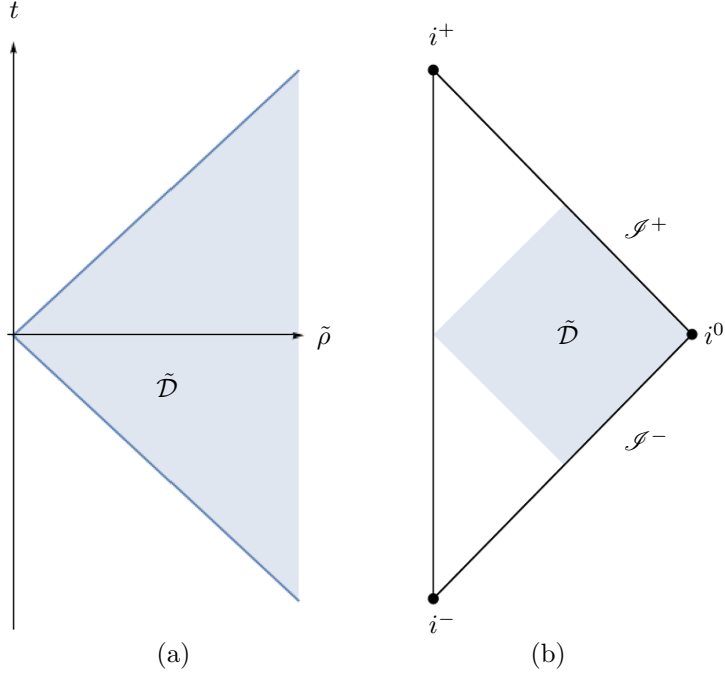

\subsection{The cylinder at spatial infinity}
\label{Section:MinkowskiCylinder}
A different representation of spatial infinity can be obtained from
the rescaling
\begin{eqnarray}
&& \bar{\bmeta} \equiv \frac{1}{\rho^2}\hat{\bmeta}. \label{cylinder-pointCompactified-conformal-relation}
\end{eqnarray}
In this representation it is convenient to introduce a new
time coordinate $\tau$ via the relation 
\begin{equation}
t = \rho \tau,
\label{Definition:TauMinkowski}
\end{equation}
so that 
\begin{equation}
\bar{\bmeta} = \mathbf{d}\tau \otimes \mathbf{d}\tau +
   \frac{\tau}{\rho}\big( \mathbf{d}\tau\otimes \mathbf{d}\rho +
   \mathbf{d}\rho \otimes \mathbf{d}\tau \big) -
   \frac{(1-\tau^2)}{\rho^2}\mathbf{d}\rho\otimes\mathbf{d}\rho -\bmsigma \label{Minkowski:Cylinder},
\end{equation}
with the contravariant metric given by
\[
\bar{\bmeta}^\sharp =(1-\tau^2) \bmpartial_\tau \otimes \bmpartial_\tau +
\rho\tau\big(\bmpartial_\tau\otimes\bmpartial_\rho +
\bmpartial_\rho\otimes\bmpartial_\tau) -\rho^2
\bmpartial_\rho\otimes\bmpartial_\rho -\bmsigma^\sharp.
\]

Introducing the coordinate $\varrho \equiv - \ln \rho$, the metrice
$\bar{\bmeta}$ can be rewritten as 
\[
\bar{\bmeta} =\mathbf{d}\tau\otimes \mathbf{d}\tau
-(1-\tau^2)\mathbf{d}\varrho\otimes\mathbf{d}\varrho -\tau \big(
\mathbf{d}\tau\otimes\mathbf{d}\varrho + \mathbf{d}\varrho\otimes\mathbf{d}\tau  \big)-\bmsigma.
\]
From this metric and observing equation \eqref{cylinder-pointCompactified-conformal-relation} and \eqref{pointCompactified-metric}, one can readily see that spatial infinity $i^0$, which corresponds to the condition $\rho=0$, lies at an infinite distance as measured by the metric $\bar{\bmeta}$.

\medskip
It follows from the previous discussion that one can write
\begin{equation}
\bar{\bmeta} =\Theta^2 \tilde{\bmeta}, \qquad \Theta\equiv \rho
(1-\tau^2).
\label{MinkowskiFGaugeHorizontal}
\end{equation}
Consistent with the above one can define the set
\[
\bar{\mathcal{M}}\equiv \big\{ p\in \mathbb{R}^4 \;|\; -1\leq \tau(p) \leq
1,\; \rho(p)\geq 0  \big\},
\]
which gives rise to a conformal extension
$(\bar{\mathcal{M}},\bar{\bmeta})$ of Minkowski spacetime. In this
representation, the following sets play an important role in our discussion:
\begin{equation*}
 I \equiv \big\{ p \in \bar{\mathcal{M}} \; \rvert   \;\;  |\tau(p)|<1, \; \rho(p)=0\big\}, 
\qquad I^{0} \equiv \big\{ p \in \bar{\mathcal{M}}\; \rvert \;
  \tau(p)=0, \; \rho(p)=0\big\},
\end{equation*} 
and 
\begin{equation*}
 I^{+} \equiv \big\{ p\in \bar{\mathcal{M}} \; \rvert \; \tau(p)=1, \; \rho(p)=0
  \big\}, \qquad I^{-} \equiv \big\{p \in \bar{\mathcal{M}}\; \rvert \; \tau(p)=-1, \; \rho(p)=0\big\}.
\end{equation*}
Moreover, future and past null infinity are given by:
\[
\mathscr{I}^\pm \equiv \big\{ p\in \bar{\mathcal{M}} \;\rvert \;
\tau(p) =\pm 1 \big\}.
\]

\begin{remark}
{\em In the following, we call the above conformal representation of the neighbourhood of spatial infinity the \emph{F-gauge horizontal representation}.
}
\end{remark}

\begin{remark}
\label{Remark:EinsteinUniverse}
{\em Although the metric \eqref{Minkowski:Cylinder} is singular on
  $I$, it induces the Lorentzian 3-metric
\begin{equation}
\bar\bmq = \mathbf{d}\tau\otimes \mathbf{d}\tau-\bmsigma,
\label{EinsteinUniverseMetric}
\end{equation}
which is regular. This metric can be regarded as the $1+2$-dimensional
version of Einstein's universe metric. In particular, its Ricci
tensor is proportional to the metric ---i.e. one has an Einstein space.
 }
\end{remark}

\subsubsection{Conformal geodesics}
A central aspect of the conformal representation of Minkowski
spacetime given by $(\bar{\mathcal{M}},\bar{\bmeta})$ is its relation
to conformal geodesics. More precisely, one has the following:

\begin{lemma}
\label{Lemma:CGMinkowski}
The pair $(x(s),\bar{\bmbeta}(s))$, $s\in[-1,1]$ with
\begin{equation}
x(s)=(s,\rho_\star, \theta_\star^{\mathcal{A}}), \qquad \bar{\bmbeta} =
\frac{1}{\rho_\star} \mathbf{d}\rho,
\label{SolutionsCGMinkowskiFgauge}
\end{equation}
for fixed $(\rho_\star, \theta_\star^{\mathcal{A}})\in\mathcal{S}_\star$
constitutes a non-intersecting congruence of conformal geodesics in
$\bar{\mathcal{M}}$. 
\end{lemma}

The details of the calculations providing the proof of the above lemma
can be found in Appendix \ref{Appendix:CGMinkowski}.

\begin{remark}
{\em In particular, the pair as given by the relations in
  \eqref{SolutionsCGMinkowskiFgauge} is a solution to the
  $\bar\bmg$-conformal geodesic equations.}
\end{remark}

\begin{remark}
{\em  In the following, in a slight abuse of notation, we identify the
affine parameter $s$ of the conformal geodesics with the time coordinate
$\tau$. The conformal geodesics given by Lemma \ref{Lemma:CGMinkowski}
have tangent vector given by $\bmpartial_\tau$.}
\end{remark}

The conformal factor $\Theta$ given in equation
\eqref{MinkowskiFGaugeHorizontal} can be deduced from the solution
to the conformal geodesic equation given by Lemma
\ref{Lemma:CGMinkowski} and Proposition
\ref{canconical-confromalFactor}. More precisely, writing the canonical conformal factor in terms of the parameter $s$ as
\[
\Theta(s) = \Theta_{\star} + \dot{\Theta}_{\star} s + \frac{1}{2} \ddot{\Theta}_{\star} s^2,
\]
with $\dot{\Theta}_{\star}$ and $\ddot{\Theta}_{\star}$ given by the
relations in \eqref{canconical-conformalFactor2}. From Lemma
\ref{Lemma:CGMinkowski}, it readily follows that 
\[
\dot{\Theta}_{\star}
= 0. 
\] 
Moreover, 
\[
\Theta_\star = \rho, \qquad \ddot{\Theta}_\star =-2 \rho. 
\]
Thus, to recover the conformal factor in
\eqref{MinkowskiFGaugeHorizontal} one identifies the parameters $s$ and $\tau$.

\begin{remark}
\label{Remark:ChoiceInitialConformalFactor}
{\em While on the one hand one has $\Theta_\star|_{\rho=0}=0$, on the
  other hand $\mathbf{d}\Theta|_{\rho=0}\neq 0$. Thus, the choice of
  $\Theta_\star$ in this conformal representation is not that one of a
  point compactification. More generally, if $\Omega$ is a 
  conformal factor giving rise to a point compactification of an
  asymptotic end of an asymptotically Euclidean manifold (cf. the
  conditions in \eqref{PointCompactificationConditions}) the
  prescription of $\Theta_\star$ is of the form $\kappa^{-1}\Omega$
  with $\kappa$ a smooth function of the form $\kappa =\varkappa \rho$
  and $\varkappa(i^0)=1$. 
}
\end{remark}

\subsubsection{Gauge freedom}
The conformal representation of the Minkowski spacetime described in
the previous subsections can be generalised to obtain a description in
which null infinity does not coincide with hypersurfaces of constant
$\tau$. For this, instead of relation \eqref{Definition:TauMinkowski}
one rather considers
\[
t = \kappa \tau, \qquad \kappa =\varkappa \rho, \qquad \varkappa =O(\rho^0),
\]
with $\varkappa$ a smooth function of the spatial coordinates. This
leads to the conformal factor 
\[
\Theta = \frac{\rho}{\varkappa} \big( 1-\varkappa^2 \tau^2\big),
\]
with associated metric $\bar\bmeta =\Theta^2 \tilde\bmeta$
\begin{equation}
\bar\bmeta = \mathbf{d}\tau\otimes\mathbf{d}\tau + \frac{\tau
  \kappa'}{\kappa}(\mathbf{d}\tau\otimes\mathbf{d}\rho +
\mathbf{d}\rho\otimes\mathbf{d}\tau) -
\frac{1}{\kappa^2}(1-\tau^2\kappa^{\prime 2})
\mathbf{d}\rho\otimes\mathbf{d}\rho -\frac{\rho^2}{\kappa^2}\bmsigma,
\qquad \kappa' \equiv \frac{\partial \kappa}{\partial \rho}.
\label{Minkowski:FGaugeGeneral}
\end{equation}
In this case, the neighbourhood of spatial infinity is given by
\[
\bar{\mathcal{M}}\equiv \big\{ p\in \mathbb{R}^4 \;|\; -\frac{1}{\varkappa(p)}\leq \tau(p) \leq
\frac{1}{\varkappa(p)},\; \rho(p)\geq 0  \big\},
\]
while null infinity is described by the sets
\[
\mathscr{I}^\pm \equiv \bigg\{ p\in \bar{\mathcal{M}} \;\rvert \;
\tau(p) =\pm \frac{1}{\varkappa(p)} \bigg\}.
\]

\begin{remark}
{\em The general F-gauge representation of Minkowski spacetime is
  related to the horizontal representation associated with the line
  element \eqref{Minkowski:Cylinder} via a conformal transformation
  which is the identity on $I$. Moreover, also notice that the
  parameter of the conformal geodesics in both representations do not
  coincide and are related to each other via a reparametrisation
  (M\"obius transformation). }
\end{remark}

\begin{remark}
{\em A key difference between the horizontal representation line
  element \eqref{Minkowski:Cylinder} and the more general F-gauge line
element \eqref{Minkowski:FGaugeGeneral} is that the former metric
loses rank (i.e. degenerates) all along null infinity ($\mathscr{I}^+$) while the
latter does it only at the critical sets ($I^\pm$). Thus, horizontal
representation provides a slightly more singular representation of the neighbourhood of spatial infinity. This singular behaviour does not
play a role in the subsequent discussion of this article. Accordingly,
given its relative analytic simplicity, all further discussions of the
F-gauge are done in the horizontal representation.}
\end{remark}

\subsection{The hyperboloid at spatial infinity}
\label{Section:UnitHyperboloid}

\begin{figure}[t]
\centering
\includegraphics[width=120mm]{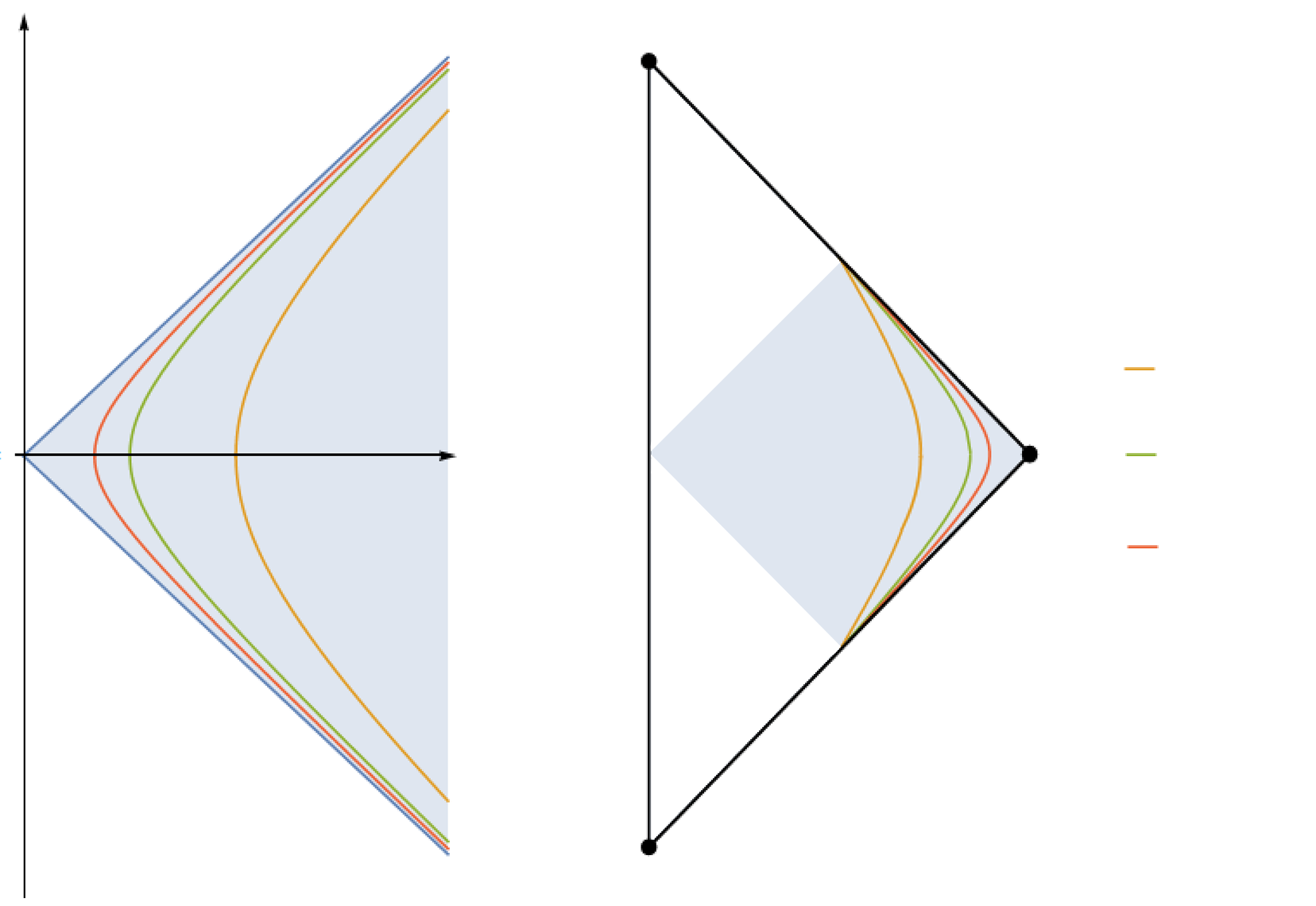}
\put(-30,137){$\psi=2$}
\put(-30,115){$\psi=4$}
\put(-30,93){$\psi=6$}
\put(-215,115){$r(t)$}
\put(-337,240){$t$}
\put(-67,115){$i^0$}
\put(-175,228){$i^+$}
\put(-175,1){$i^-$}
\put(-125,50){$\mathscr{I}^-$}
\put(-120,170){$\mathscr{I}^+$}
\put(-128,-10){(b)}
\put(-280,-10){(a)}
\caption{(a) The timelike hyperboloids in the domain $\tilde{\mathcal{D}}$ of the Minkowski spacetime used in the construction of Ashtekar's hyperboloid at spatial infinity. (b) The timelike hyperboloids shown on the conformal diagram of Minkowski spacetime.}
\label{Fig:timelike-hyperboloids}
\end{figure}
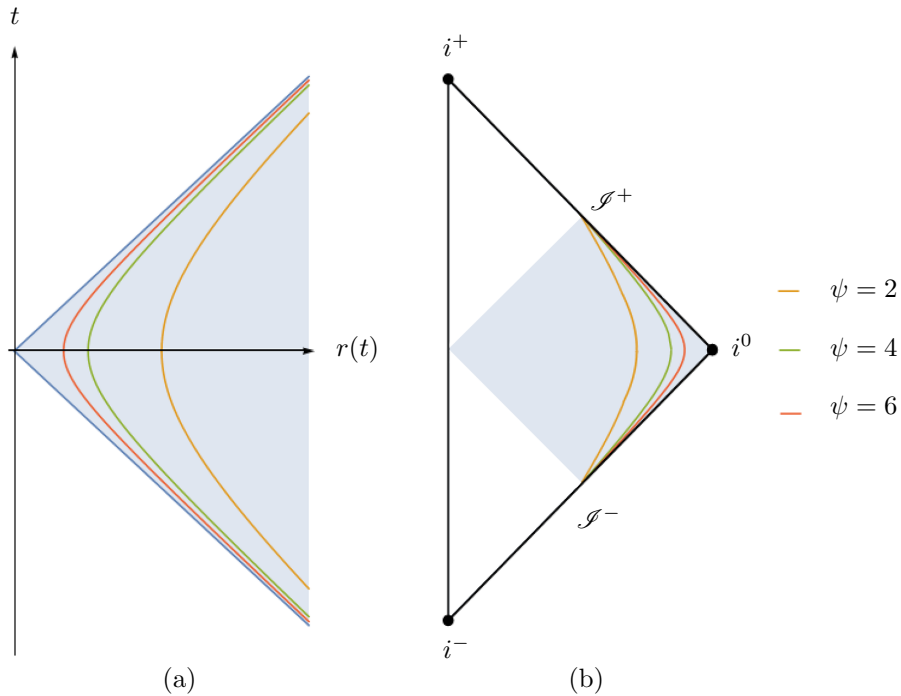

Following the discussion in \cite{AshRom92} a different (albeit
related) representation of spatial infinity as an extended set can be
obtained by considering \emph{hyperbolic coordinates}
$(\psi,\chi,\theta,\varphi)$, related to the standard Cartesian coordinates $(x^{\mu})$ via the relations
\begin{subequations}
\begin{eqnarray}
&& x^0 = \psi \sinh \chi, \label{HyperbolicCoords1}\\
&& x^1 = \psi \cosh \chi \sin\theta \cos\varphi, \label{HyperbolicCoords2}\\
&& x^2 = \psi \cosh \chi \sin \theta \sin \varphi, \label{HyperbolicCoords3}\\
&& x^3 = \psi \cosh \chi \cos \theta. \label{HyperbolicCoords4}
\end{eqnarray}
\end{subequations}
It readily follows that 
\[
\tilde{\rho}^2 -\tilde{t}^2 =\psi^2.
\]
Accordingly, the hypersurfaces of constant $\psi$ are timelike hyperboloids ---see Figure \ref{Fig:timelike-hyperboloids}. As in Section
\ref{Section:MinkowskiPointSpatialInfinity} one has that 
\[
X^2 = \eta_{\mu\nu}x^\mu x^\nu = \tilde{t}^2-\tilde{\rho}^2 =- \psi^2,
\]
so that, as before, spatial infinity is contained in the region
\[
\tilde{\mathcal{D}} \equiv \{ p \in \mathbb{R}^{4} \; | \; X^2(p) <0  \}. 
\]
It follows that the hyperbolic coordinates only cover the domain
$\tilde{\mathcal{D}}$ ---see Figure \ref{Fig:Domain-D}. In hyperbolic coordinates, one can write $\tilde{\bmeta}$ as
\[
\tilde{\bmeta}= - \mathbf{d}\psi \otimes \mathbf{d}\psi +
   \psi^2 \bmell,
\]
where 
\begin{equation}
\bmell \equiv \mathbf{d}\chi\otimes \mathbf{d}\chi -\cosh^2\chi \bmsigma,
\label{hyperboloicmetric}
\end{equation}
is the (Lorentzian) metric of the \emph{unit hyperboloid}. This metric
can be obtained from the pullback of $\tilde{\bmeta}$ to the
hypersurface defined by the condition
$\tilde{\rho}^2-\tilde{t}^2=1$. 

\subsubsection{Standard representation of spatial infinity}
In order to discuss the behaviour near
spatial infinity it is convenient to introduce a new coordinate
\[
\zeta \equiv \frac{1}{\psi}.
\] 
It follows then that 
\begin{equation}
\tilde{\bmeta} = -\frac{1}{\zeta^4}\mathbf{d}\zeta \otimes
\mathbf{d}\zeta + \frac{1}{\zeta^2}\bmell.
\label{Minkowski:ZetaHyperbolic}
\end{equation}
Thus, it seems natural to introduce a conformal factor of the form
\[
\Omega \equiv \zeta^2,
\]
so as to obtain
\begin{eqnarray*}
&& \hat{\bmeta} \equiv \Omega^2 \tilde{\bmeta}\\
&& \phantom{\hat{\bmeta}} = - \mathbf{d}\zeta \otimes
\mathbf{d}\zeta + \zeta^2\bmell.
\end{eqnarray*}
This leads to the \emph{standard} representation of spatial infinity
as a point. This is shown
by noting that the metric of the 2-surfaces of constant $\zeta$ is
given by $\zeta^2 \bmell$. Hence, the set given by $\zeta=0$ has  zero
volume and is forced to be a single point by the choice of conformal
factor $\Omega = \zeta ^2$. Moreover, one can confirm that at $\zeta
=0$, we have $\Omega=0$, $\mathbf{d}{\Omega} =0$, $\mbox{Hess}\,\Omega = -2 \hat{\bmeta}$.

\subsubsection{The hyperboloid at spatial infinity}
\label{Subsubsection:MinkowskiUnitHyperboloid}
In the present case, it is better to define the conformal factor
\[
H \equiv \zeta,
\]
so that 
\begin{eqnarray}
&& {\bmeta} \equiv H^2 \tilde{\bmeta}, \nonumber\\
&& \phantom{\check{\bmeta}} = -\frac{1}{\zeta^2}\mathbf{d}\zeta
   \otimes\mathbf{d}\zeta + \bmell. \label{ConformalMinkowski:AshthekarHyperboloid}
\end{eqnarray}
This particular rescaling readily connects with Friedrich's framework
already discussed in Section \ref{Section:Relating-Ashtekar-F-Gauge-Minkowski}.

\medskip
Consider in the following the timelike hyperboloids defined by the condition
$\zeta=\zeta_\bullet$ where $\zeta_\bullet$ is a constant. 
The key observation in \cite{AshRom92} is that although the conformal
metric \eqref{ConformalMinkowski:AshthekarHyperboloid} is singular at
$\zeta=0$,  
 the conformal 3-metric
\[
{\bmq} = H^2 \tilde{\bmq} =\bmell,
\]
is well defined. Observe also, that the contravariant 4-dimensional metric
\[
{\bmeta}^\sharp = -\zeta^2 \bmpartial_\zeta \otimes
\bmpartial_\zeta + \bmell^\sharp,
\]
is well defined ---although, it losses rank at the set where $\zeta=0$. In addition to the above, observe that
\[
(\mathbf{d}\zeta)^\sharp = \tilde{\bmeta}^\sharp( \mathbf{d}\zeta) =
-\zeta^4\bmpartial_\zeta.
\]
The later suggests introducing a \emph{rescaled unit normal vector}
through the relation
\[
{\bmN} = H^{-4} (\mathbf{d}\zeta)^\sharp = - \bmpartial_\zeta.
\]

\begin{remark}
{\em Starting directly from the (singular) metric
  $\bmeta$ one readily finds that the unit normal covector is
  given by $\zeta^{-1}\mathbf{d}\zeta$. Moreover, one has that 
\[
\big( \zeta^{-1}\mathbf{d}\zeta  \big)^\sharp =-\zeta \bmpartial_\zeta,
\]
which, despite being well defined at $\zeta=0$ vanishes. As it will be
seen in the following, vectors with this behaviour at infinity play a key role in
Friedrich's framework of spatial infinity.}
\end{remark}

\begin{remark}
{\em In the following, the timelike hyperboloid $\mathcal{H}$ described by the condition
$\zeta=0$ together with the induced metric $\bmq=\bmell$ will be known as \emph{the hyperboloid at spatial infinity of Minkowski spacetime}. This hyperboloid is an example of the general
notion of \emph{asymptote at spatial infinity} introduced in the
definition of an asymptotically Minkowskian spacetime at spatial
infinity (AMSI) ---see Definition \ref{Definition:AMSI} in Section \ref{Section:AMSI}. }
\end{remark}


\subsection{Relating the Ashtekar and F-gauge construction in the Minkowski spacetime}
\label{Section:Relating-Ashtekar-F-Gauge-Minkowski}

In this subsection, we obtain the explicit relation between the
description of spatial infinity in terms of the hyperboloidal
coordinates and that based on the F-gauge. The procedure for relating
these two representations will serve as a template for an analogous
computation in more general classes of spacetimes.

\subsubsection{From the hyperboloid at infinity to the cylinder at
  infinity}
In order to relate the pair $(\mathcal{H},\bmell)$ corresponding to
Ashtekar's hyperboloid at spatial infinity to the pair $(I,\bar\bmq)$
with $\bar\bmq$ as given by equation \eqref{EinsteinUniverseMetric}, it
is observed while the former can be thought of as a $1+2$ version of
the \emph{de Sitter spacetime}, the latter is a $1+2$ version of the
Einstein static Universe ---see Remark
\ref{Remark:EinsteinUniverse}. As both metrics $\bmell$ and $\bar\bmq$
are Einstein, their Cotton tensor vanishes and thus, they are
conformally flat. Accordingly, $\bmell$ and $\bar\bmq$ are conformally
related. 

\medskip
In order to find the conformal factor relating $\bmell$ and $\bar\bmq$, we follow the same procedure used to show that the (4-dimensional) de Sitter spacetime
can be conformally embedded in the (4-dimensional) Einstein static
Universe ---see e.g. Section 6.3 in \cite{CFEBook}. More precisely,
starting from the metric $\bmell$ of the unit hyperboloid, introduce
the coordinate transformation given by the relation
\[
\mathbf{d}\chi = \cosh \chi \mathbf{d}\tau.
\]
It follows readily that 
\[
\bmell = \cosh^2 \chi \big( \mathbf{d}\tau \otimes \mathbf{d}\tau -\bmsigma  \big),
\]
from where one can indeed see that  $\bmell$ and $\bar\bmq$ are
conformally related. In particular, it can be shown that $\cosh \chi
=\sec \tau$. 

\begin{remark}
{\em It follows from the previous discussion that Friedrich's cylinder
at spatial infinity is indeed a time compactified version of
Ashtekar's hyperboloid. In particular, the critical sets $I^\pm$
correspond to the limits $\chi\rightarrow \pm \infty$.}
\end{remark}

\subsubsection{Relating the neighbourhoods of spatial infinity}
\label{Subsection:RelatingSpatialInfinityMinkowski}
Now, to relate the two constructions away from spatial infinity we recall that
\begin{equation}
 \bar{\bmeta} = \Theta^2 \tilde\bmeta, \qquad {\bmeta} =H^2
 \tilde\bmeta.
\label{ConformalTransformationsMinkowski}
\end{equation}
Now, as $\zeta=1/\psi$, a direct computation then
gives that 
\begin{eqnarray}
&& \rho = \frac{1}{\psi}\cosh\chi, \qquad \tau =\tanh \chi, \label{coodrinate-transformation}
\end{eqnarray}
so that
\[
\zeta = \rho\sech\chi.
\]
Moreover, from the rescalings in
\eqref{ConformalTransformationsMinkowski} it follows that 
\[
\bmeta = \varpi^2 \bar\bmeta, \qquad \varpi \equiv H\Theta^{-1}.
\]
In terms of coordinates, one has
\[
\varpi = \frac{\zeta}{\rho(1-\tau^2)} =\cosh\chi.
\]
Observing that $\chi =\arctanh \tau$ it follows that
\[
\varpi = \frac{1}{\sqrt{1-\tau^2}}.
\]
Notice that $\varpi|_{\mathcal{S}_\star}=1$. Moreover,
$\varpi\rightarrow \infty$ as $\tau\rightarrow\pm 1$. Thus, $\varpi$
gives rise to a conformal representation of Minkowski
spacetime that does not include null infinity.

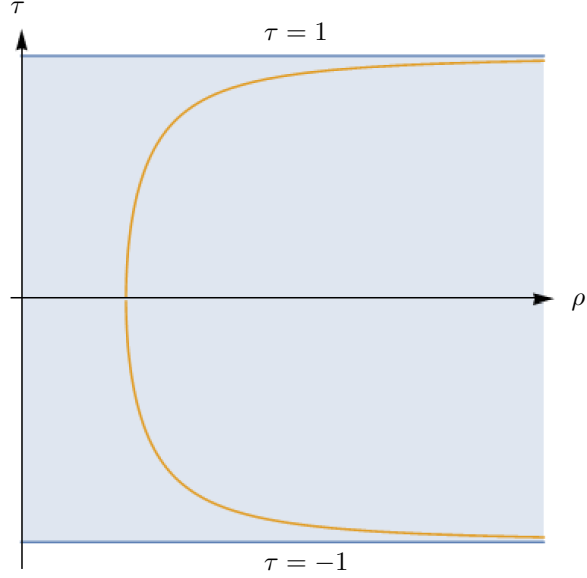
\begin{figure}[t]
\centering
\includegraphics[width=80mm]{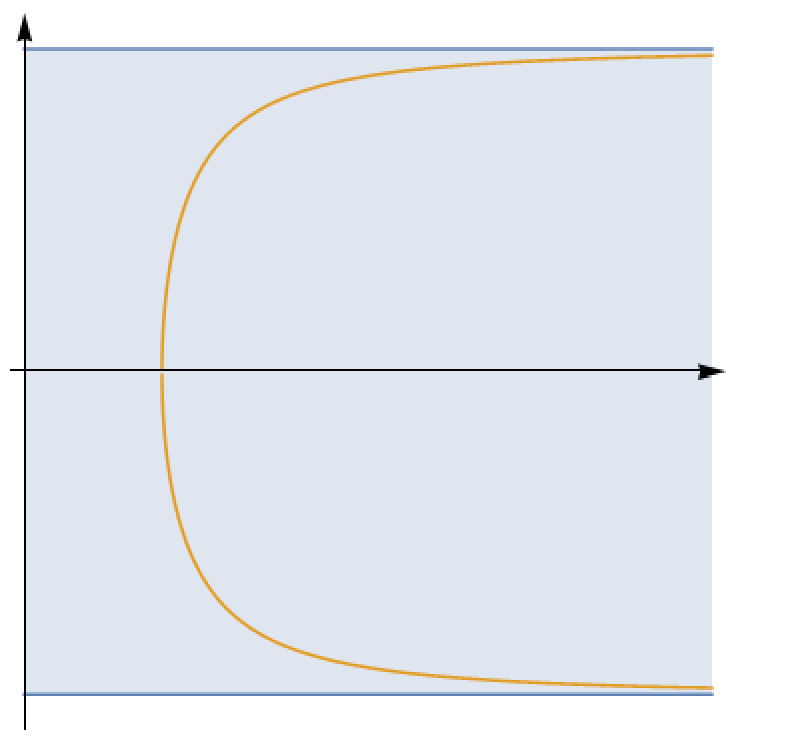}
\put(-130,205){$\tau=1$}
\put(-130,5){$\tau=-1$}
\put(-225,215){$\tau$}
\put(-15,105){$\rho$}
\caption{Example of one of the timelike hyperboloids used in
  Ashtekar's representation of spatial infinity for the Minkowski as
  seen from the point of view of the F-gauge. Observe that the
  hyperboloid asymptotes the sets described by the conditions
  $\tau=\pm 1$. Notice, however, that in this description the
  hyperboloid at spatial infinity is compact and corresponds to the
  portion of the vertical axis between the values $\tau=-1$ and
  $\tau=1$. }
\label{Fig:constant-psi-hyperboloids}
\end{figure}

\begin{remark}
{ \em Using equation (\ref{coodrinate-transformation}) and setting $\psi=\psi_\bullet$ where $\psi_\bullet$ is a constant. One has that $(\tanh \chi, \psi_\bullet^{-1}\cosh\chi)$ describes a curve in the $(\tau,\rho)$-plane ---see the Figure \ref{Fig:constant-psi-hyperboloids}. Observe, in particular, that 
\[
 (\tanh \chi, \psi_\bullet^{-1}\cosh\chi)\longrightarrow (\pm
 1,\infty) \qquad \mbox{as} \qquad \chi\rightarrow \pm \infty.
\]
Accordingly, the hyperboloids of constant $\psi$ never reach the
conformal boundary ---i.e. $\mathscr{I}^\pm$; see Figure
\ref{Fig:constant-psi-hyperboloids}}. Hence, this representation of
spatial infinity cannot be used to study the effects of gravitational
radiation and the relation between asymptotic charges at null infinity
and conserved quantities at spatial infinity.
\end{remark}

\subsubsection{Conformal geodesics and construction of a conformal Gaussian system}

The unphysical metric $\bar{\bmeta}$ is conformally related to the
physical metric $\tilde{\bmeta}$ via the conformal factor $\Theta=
\rho (1-\tau^2)$. Thus, a solution to the conformal geodesic equations
on $\bar\bmeta$ implies a solution to the conformal geodesic equation on
the physical metric $\tilde{\bmeta}$. Recalling that the pair given by
$x(\tau) = (\tau, \rho_{\star}, \theta^{\mathcal{A}}_{\star})$ and
$\bar\bmbeta =\mathbf{d}{\rho}/\rho$ satisfies the conformal geodesic
equations for the metric $\bar\bmeta$, it follows, using the coordinate
transformation to hyperbolic coordinates given by
\eqref{HyperbolicCoords1}-\eqref{HyperbolicCoords4}, that the tangent
vector to the curve $x(\tau)$ is given in the
$(\zeta,\chi,\theta^A)$ coordinates by
\[
\dot{\bmx} = -\frac{\rho_{\star} \tau}{\sqrt{1-\tau^2}} \bmpartial_{\zeta} + \frac{1}{1-\tau^2} \bmpartial_{\chi}.
\]
Then, $(\dot{\bmx}, {\bmbeta})$ provides a solution to the
conformal geodesic equation, where ${\bmbeta}$ is given by
\begin{eqnarray*}
&& {\bmbeta} = \bar\bmbeta - \varpi^{-1} \mathbf{d}\varpi \\
&& \phantom{\bmbeta} =
   \frac{1}{\zeta} \mathbf{d}\zeta =\mathbf{d}\ln H, 
\end{eqnarray*}
where it has been used that 
\[
\bar\bmbeta = \frac{1}{\zeta}\mathbf{d}\zeta +\tanh \chi \mathbf{d}\chi,
\qquad \mathbf{d}\ln \varpi = \tanh\chi \mathbf{d}\chi.
\]

\medskip
\noindent
\textbf{Setting up a conformal Gaussian system.} To conclude this discussion, we show how to compute a conformal Gaussian system on top of Ashtekar's conformal representation of
spatial infinity in Minkowski spacetime. As we have already shown that we have a non-intersecting congruence of conformal geodesics, one can invoke Proposition
\ref{canconical-confromalFactor} in Section \ref{Section:CG} so that one can associate a
conformal factor 
\[
\Lambda (s) = {\Lambda}_\star
+\dot{\Lambda}_\star s +\frac{1}{2}\ddot{\Lambda}_\star s^2,
\]
along each of the curves of the congruence and where $s$ is the
natural 
parameter of the curves. The coefficients in the
above expression satisfy the relations
\[
 \dot{\Lambda}_\star = \langle \tilde\bmbeta_\star, \dot\bmx_\star\rangle\Lambda_\star, \qquad
 \Lambda_\star\ddot{{\Lambda}}_\star =
 \frac{1}{2}\tilde{\bmeta}^\sharp
 (\tilde\bmbeta_\star,\tilde\bmbeta_\star). 
\]

Consistent with the conformal metric
\eqref{ConformalMinkowski:AshthekarHyperboloid}, for $\Lambda_\star$
one prescribes $\Lambda_\star =\zeta$. For the covector $\tilde\bmbeta$ one would like to prescribe a behaviour of the form
$\tilde\bmbeta_\star=2\mathbf{d}\tilde{\rho}/\tilde{\rho}$. Note that the
above form of $\tilde{\bmbeta}$ is consistent with the conformal geodesic
solution obtained on the F-gauge construction with $\bar{\bmbeta}
\propto 1/\rho$. Also, this type of behaviour near spatial infinity is obtained from the
conformal factor $\Omega=\zeta^2=1/\psi^2\sim1/r^2$, which renders the
\emph{standard} point compactification of spatial infinity. Consistent with this
discussion set
\[
\tilde\bmbeta =\mathbf{d}\ln \Omega =\frac{2}{\zeta}\mathbf{d}\zeta. 
\]
it follows from the above discussion that $\dot{{\Lambda}}_\star =0$
as $\langle \mathbf{d}\zeta,\bmpartial_\tau\rangle =0$. Also, notice
that
\[
\tilde{\bmeta}^\sharp(\tilde\bmbeta_\star,\tilde\bmbeta_\star) = \tilde\bmh^\sharp(\tilde\bmbeta_\star,\tilde\bmbeta_\star),
\]
where $\tilde{\bmh}$ is the metric on hypersurfaces of constant $\chi$. A quick computation gives that 
\[
\tilde{\bmh} =-\frac{1}{\zeta^4}\mathbf{d}\zeta\otimes\mathbf{d}\zeta
-\frac{1}{\zeta^2}\bmsigma, \qquad \tilde{\bmh}^\sharp = -\zeta^4
\bmpartial_\zeta\otimes\bmpartial_\zeta +\zeta^2 \bmsigma^\sharp,
\]
so that $\tilde\bmh^\sharp(\tilde\bmbeta_\star,\tilde\bmbeta_\star)
=-4\zeta^2$. From the latter it follows then that $\ddot{\Lambda}_\star =-2\zeta$. Accordingly, the conformal factor associated to the congruence of conformal
geodesics has the form
\[
\Lambda= \zeta(1-s^2). 
\]
Finally, defining $\breve\bmeta \equiv \Lambda^2 \tilde\bmeta$
 so that 
\[
\breve\bmbeta =\tilde\bmbeta -\mathbf{d}\ln \Lambda,
\]
one concludes that $\breve{\bmbeta}_\star =
\mathbf{d}\zeta/\zeta$. This shows the consistency of the prescription
of initial data for the covector $\tilde\bmbeta$ given above.

\section{Asymptotically Minkowskian spacetimes at spatial infinity}
\label{Section:AMSI}

In \cite{AshRom92} the notion of spacetimes which are
\emph{asymptotically Minkowskian at spatial infinity} has been introduced. We want to analyse this definition in light of Friedrich's framework. Hence, we briefly review the relevant definitions and properties.

\subsection{Definitions}
Following the discussion in \cite{AshRom92}, in the following we will consider the following definition:
\begin{definition}
\label{Definition:AMSI}
A vacuum spacetime $(\tilde{\mathcal{M}},\tilde{\bmg})$ is said to \textbf{possess an asymptote at spatial infinity} if there exists a manifold with boundary
$\mathcal{H}$, a smooth function $\Omega$ defined on $\mathcal{M}$ and a diffeomorphism from $\tilde{\mathcal{M}}$ to $\mathcal{M}\setminus\mathcal{H}$ (which is used to identify
$\tilde{\mathcal{M}}$ with its image in $\mathcal{M}$; in particular $\Omega$ is non-vanishing in $\tilde{\mathcal{M}}$) such that:
\begin{itemize}
    \item[(i)] $\Omega = 0$ and $\mathbf{d}\Omega \neq 0$ on $\mathcal{H}$;
    \item[(ii)] the fields
    \[
    \bmq \equiv \Omega^2\big( \tilde{\bmg} + \Omega^{-4}
    \mathbf{d}\Omega\otimes \mathbf{d}\Omega\big)
    \]
    and 
    \[
    \bmN \equiv \Omega^{-4} \tilde{\bmg}^\sharp(\mathbf{d}\Omega,\cdot)
    \]
    admit smooth limits to $\mathcal{H}$. In particular, the pullback of
    $\bmq$ (to be denoted again by $\bmq$) to $\mathcal{H}$ is also well defined and has signature $(+--)$;
    \item $\lim_{\Omega \to 0} \Omega^{-1} \tilde{G}_{ab} =0$ on $\mathcal{H}$.
\end{itemize} 
In addition, if $\mathcal{H}$ has the topology of $\mathbb{R}\times\mathbb{S}^2$ then $(\tilde{\mathcal{M}},\tilde{\bmg})$ is said to be \textbf{asymptotically flat at spatial infinity}. Moroever, if $\mathcal{H}$ is geodesically complete with respect to $\bmq$ we say that the spacetime is \textbf{asymptotically Minkowskian at spatial infinity (AMSI).} 
\end{definition}

\medskip
\noindent
\textbf{Notation.} In the following, for convenience, we will make use of the symbol $\simeq$ to denote equality on $\mathcal{H}$. With this notation the conditions in point \emph{(i)} of the Definition \ref{Definition:AMSI} are written as
\[
\Omega \simeq 0, \qquad \mathbf{d}\Omega\not\simeq 0.
\]

\begin{remark}
{\em The above definition involves some conformal gauge freedom in the sense that if $(\mathcal{M}, \Omega)$ satisfy Definition \ref{Definition:AMSI} and $\alpha$ is a smooth function which is a non-zero constant  on $\mathcal{H}$ and non-vanishing in $\tilde{\mathcal{M}}$, then $(\mathcal{M}, \Omega'=\alpha \Omega)$ also satisfy the definition.
} 
\end{remark}

\subsection{Properties and consequences of the definition of an AMSI spacetime}
In the following, we explore some direct consequences of Definition \ref{Definition:AMSI}, which will be used repeatedly in the rest of this article. Accordingly, throughout, we assume that one has a vacuum
AMSI spacetime  $(\tilde{\mathcal{M}},\tilde{\bmg})$. 

\medskip
For conceptual clarity, let $\zeta$ denote a coordinate such that in a neighbourhood of $\mathcal{H}$ one has
\[
\Omega =\zeta. 
\]
It follows then that the metric 
\[
\bmq \equiv \zeta^2\big( \tilde{\bmg} + \zeta^{-4}
\mathbf{d}\zeta\otimes \mathbf{d}\zeta\big),
\]
has a smooth limit as $\zeta\rightarrow 0$. From this assumption, it follows that $\bmq$ can be written as
\begin{equation}
\bmq =\mathring{\bmq} +\breve{\bmq},
\label{qcloseH}
\end{equation}
with $\mathring{\bmq}$ independent of $\zeta$ and, $\mathring{\bmq}(\bmpartial_\zeta,\cdot)=0$. The tensor $\breve{\bmq}$ has smooth components such that 
\[
\breve{\bmq}=o(\zeta).
\]

\begin{remark}
{\em Recall that the notation $f(\zeta) =
  o(\zeta^\alpha)$ means that $f(\zeta)/\zeta^\alpha\rightarrow 0$ as
  $\zeta\rightarrow 0$.}
\end{remark}
The statement \eqref{qcloseH} can be regarded as a \emph{zeroth-order} Taylor expansion of the metric $\bmq$ with respect to $\zeta$. From the above, it follows that the physical metric $\tilde{\bmg}$ has, close to $\mathcal{H}$, the form
\[
\tilde{\bmg} = -\frac{1}{\zeta^4} \mathbf{d}\zeta \otimes
  \mathbf{d}\zeta + \frac{1}{\zeta^2} \big(\mathring{\bmq} +\breve{\bmq}\big).
\]
Following the analogy of Minkowski spacetime, define the conformal metric 
\[
{\bmg} \equiv \Omega^2 \tilde{\bmg}, 
\]
so that one has
\begin{equation}
{\bmg} = -\frac{1}{\zeta^2} \mathbf{d}\zeta\otimes
\mathbf{d}\zeta + \big( \mathring{\bmq}+\breve{\bmq}),
\label{gcloseH}
\end{equation}
---cf. the line element \eqref{ConformalMinkowski:AshthekarHyperboloid} for the conformal Minkowski spacetime in hyperboloidal coordinates.
\begin{remark}
{\em As in the case of Minkowski spacetime,  the conformal metric \eqref{gcloseH} is singular at $\zeta=0$. Dealing with this singular behaviour will be the main challenge in the subsequent analysis of this section. }
\end{remark}
\begin{remark}
{\em As it will be seen, the metric \eqref{gcloseH} can be further specialised by choosing coordinates on $\mathcal{H}$ so that 
\[
\mathring{\bmq} =\bmell = \mathbf{d}\chi\otimes \mathbf{d}\chi - \cosh^2
\chi \bmsigma,
\]
the metric of the unit timelike hyperboloid.}
\end{remark}

\subsection{The conformal constraint equations on $\mathcal{H}$}
In this section, we discuss the implications of Definition \ref{Definition:AMSI} on the conformal Einstein constraint equations \eqref{ConformalConstraint1} \eqref{ConformalConstraint8} when evaluated on the asymptote $\mathcal{H}$. In this way, we systematically recover the conditions satisfied by the gravitational field on the hyperboloid as discussed in~\cite{AshRom92}.

\medskip
In order to evaluate the conformal Einstein constraints on $\mathcal{H}$, we first consider the equations \eqref{ConformalConstraint1}-\eqref{ConformalConstraint8} on timelike hypersurfaces for which the conformal factor $\Omega = \zeta$ is constant. On these hypersurfaces, one has, in adapted coordinates, that $D_a \zeta=0$ and $D_a D_b \zeta=0$. Following the discussion in \cite{AshRom92}, define 
\[
\digamma\equiv  N^a \nabla_a \Omega.
\]
Note that the unit normal is related to Ashtekar's normal by $n^a = \nu N^a$ with $\nu\equiv \Omega \digamma^{- \frac{1}{2}} = O(\Omega)$ i.e. as $\zeta \rightarrow 0$, we have $|\nu| \leq M |\Omega|$, where $M$ is a positive constant. From (\ref{Sigma-definition}), we have $\Sigma = \nu \digamma$. Taking the limit of
equations \eqref{ConformalConstraint1}-\eqref{ConformalConstraint8} as $\zeta\rightarrow 0$, one obtains the following equations on $\mathcal{H}$:
\begin{subequations}
\begin{eqnarray}
&& D_a \digamma \simeq 0, \label{ConfCostraintHyp2}\\
&& D_a s \simeq 0, \label{ConfCostraintHyp8}\\
&& D_a L_{bc} - D_b L_{ac} \simeq 0 ,\label{ConfCostraintHyp3}\\
&& D_a L_{b \bot} - D_b L_{a \bot} - K_a{}^c L_{bc} + K_b{}^c L_{ac} - n_b K_a{}^c L_{c \bot} + n_a K_b{}^c L_{c \bot} \simeq 0, \label{ConfCostraintHyp4}\\
&&  D^c d_{c \bot ab} \simeq 0, \label{ConfCostraintHyp5}\\
&& D^c d_{c \bot a \bot} \simeq 0. \label{ConfCostraintHyp6} \\
&& s \simeq 0, \label{ConfCostraintHyp7}
\end{eqnarray}
\end{subequations}
As in the previous section, the symbol $\simeq$ is used to indicate an equality which holds at $\mathcal{H}$, and these equations are to be understood as the limit as $\zeta \rightarrow 0$. For example, equation \eqref{ConfCostraintHyp2} is to be understood as
$$
\lim_{\zeta \rightarrow 0} D_a \digamma =0. 
$$
The conformal field equations imply an additional constraint $D_a D_b \Omega \simeq 0$, which is satisfied identically on $\mathcal{H}$. In addition to the above relations, one can also consider the Gauss-Codazzi and Codazzi-Mainardi equations on $\mathcal{H}$. Assuming vacuum, these give the relations
\begin{subequations}
\begin{eqnarray}
&& q_a{}^c q_b{}^d L_{cd} \simeq l_{ab} + \tfrac{1}{4} q_{ab} \left( K_{cd} K^{cd} - K^2 \right) + K K_{ab} - K_{a}{}^c K_{bc}, \label{Gauss-Codazzi}\\
&& D_a K_{bc} - D_b K_{ac} \simeq 0, \label{Codazzi-Mainardi}
\end{eqnarray}
\end{subequations}
where $l_{ab}$ denotes the \emph{Schouten tensor} of the metric
$q_{ab}$ and one uses $q^{ab}$ to raise and lower indices on $\mathcal{H}$.

\medskip
A direct consequence of the above relations is that $\digamma$ is
constant on $\mathcal{H}$. Moreover, from the above relations, one obtains the following:
\begin{lemma}
\label{Lemma:TheHyperboloidIsEinstein}
For a spacetime  $(\tilde{\mathcal{M}},\tilde{\bmg})$ satisfying Definition \ref{Definition:AMSI}, it follows that on the asymptote $\mathcal{H}$, one has the relation
\[
r_{ab} \simeq 2 \digamma q_{ab}.
\]
In other words, $(\mathcal{H},\bmq)$ is an Einstein space.
\end{lemma}
\begin{proof}
Starting from the transformation rule for the Schouten tensor
\[
L_{cd} - \tilde{L}_{cd} = - \frac{1}{\Xi} \nabla_c \nabla_d \Xi + \frac{1}{2 \Xi^2} \nabla_e \Xi \nabla^e \Xi g_{cd},
\]
contracting with $q_a{}^c q_b{}^d$ and using $K_{ab} = - q_a{}^c q_b{}^d \nabla_c n_d$ and $\nabla_c \Xi \nabla^c \Xi = \Xi ^{2} \digamma $, the relation between $L_{ab}$ and $\tilde{L}_{ab}$ can be written as
\[
q_a{}^c q_b{}^d L_{cd} = q_a{}^c q_b{}^d \tilde{L}_{cd} + \digamma ^{\frac{1}{2}} K_{ab} + \frac{1}{2} \digamma q_{ab}.
\]
The first term on the right-hand side vanishes for a spacetime satisfying Definition \ref{Definition:AMSI}. Substituting into the Gauss-Codazzi equation \eqref{Gauss-Codazzi} and using $K_{ab} \simeq - \digamma^{\frac{1}{2}} q_{ab}$, we get
\[
l_{ab} = \digamma q_{ab}. 
\]
Then, defining the 3-dimensional Schouten tensor as $l_{ab} \equiv r_{ab} - \frac{1}{4} r q_{ab}$, we can confirm that $r_{ab} = 4 \digamma q_{ab}$. Making use of the conformal gauge freedom, $\digamma$ can be redefined so that $r_{ab} = 2 \digamma q_{ab}$. 
\end{proof}
In addition, one has the following \emph{peeling-type} behaviour for
the components of the Weyl tensor:
\begin{lemma}
\label{Lemma:WeylVanishesOnH}
For a spacetime  $(\tilde{\mathcal{M}},\tilde{\bmg})$ satisfying Definition \ref{Definition:AMSI}, the electric $E_{ab}$ and magnetic $B_{ab}$ parts of the Weyl tensor satisfy
\begin{eqnarray*}
&& E_{ab} \simeq 0, \\
&& B_{ab} \simeq 0.
\end{eqnarray*}
\end{lemma}

\begin{proof}
Starting with the definition of the electric and magnetic part
of the Weyl tensor
$$
E_{ab} = C_{acbd} n^c n^d, \; \; \; B_{ab} = {}^\ast C_{acbd} n^c n^d,
$$
where $^\ast C_{acbd}$ is the left Hodge dual of $C_{acbd}$. Using the
decomposition of the curvature tensor~\eqref{curvature-decomposition} and after a lengthy calculation, one
can write
\begin{eqnarray*}
&& E_{ab} = \mathcal{L}_{n} K_{ab} + \digamma^{\frac{1}{2}} D_a D_b \digamma ^{-\frac{1}{2}} - K_b{}^c K_{ac} + q_a{}^c q_b{}^d L_{cd} - q_{ab} n^c n^d L_{cd}, \\
&& B_{ab} = - \epsilon_a{}^{cd} \left( D_{[d} K_{c] b} + \frac{1}{2} \left( g_{db} L_{ec} n^e - g_{bc} L_{ed} n^e \right) \right)
\end{eqnarray*}          
where $\mathcal{L}_n$ denotes the Lie derivative in the direction of
the unit normal $n^{a}$. Then making use of the transformation law of the
Schouten tensor and the fact that $\tilde{L}_{ab}=0$ and taking the
limit as $\zeta \rightarrow 0$ and again using $K_{ab} \simeq -\digamma^{\frac{1}{2}} q_{ab} $, it can be shown that $\lim_{\zeta
\rightarrow0} E_{ab} =0$ and $\lim_{\zeta\rightarrow0} B_{ab} =0$.
\end{proof}

\subsection{Relating a general asymptote to Friedrich's cylinder}
\label{Section:HyperboloidToCylinder}
The purpose of this section is to show that, as in the case of Minkowski spacetime, the 3-dimensional Lorentzian manifold $(\mathcal{H},\bmq)$ is conformally related to the 3-dimensional Einstein Universe (Friedrich's cylinder) $(\mathbb{R}\times\mathbb{S}^2,
\bar{\bmq}\equiv \mathbf{d}\tau\otimes\mathbf{d}\tau -\bmsigma)$.

\medskip
In the following, we assume one is given an asymptote $\mathcal{H}$ as given by Definition \ref{Definition:AMSI}. From Lemma \ref{Lemma:TheHyperboloidIsEinstein}, it follows that the intrinsic Ricci tensor $r_{ij}$ can be written as
\[
r_{ij} \simeq 2 \digamma q_{ij}
\] 
where $\digamma$ is constant on $\mathcal{H}$. In order to find the conformal transformation between $\bmq$ and $\bar\bmq$, we proceed in two steps: (i) First, we show that the intrinsic metric $q_{ij}$ is conformally related to the $1+2$-dimensional Minkowski spacetime metric $\mathring{q}_{ij}$; and (ii) use the fact that the $1+2$-dimensional Minkowski spacetime is conformally related to the 3-dimensional Einstein cylinder. The composition of these two conformal rescalings gives the relation between $\bmq$ and $\bar\bmq$ ---see Figure \ref{Fig:ConformalFactors2}. 
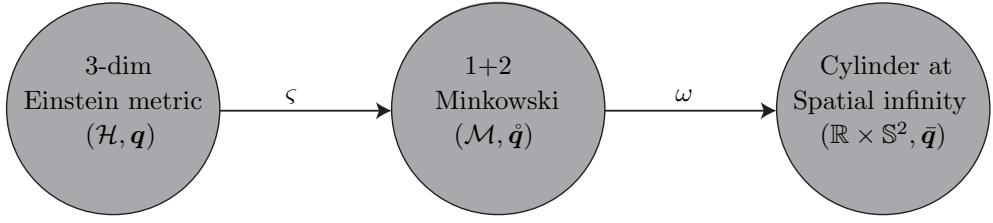
\begin{figure}[t]
\centering
\includegraphics[width=\textwidth]{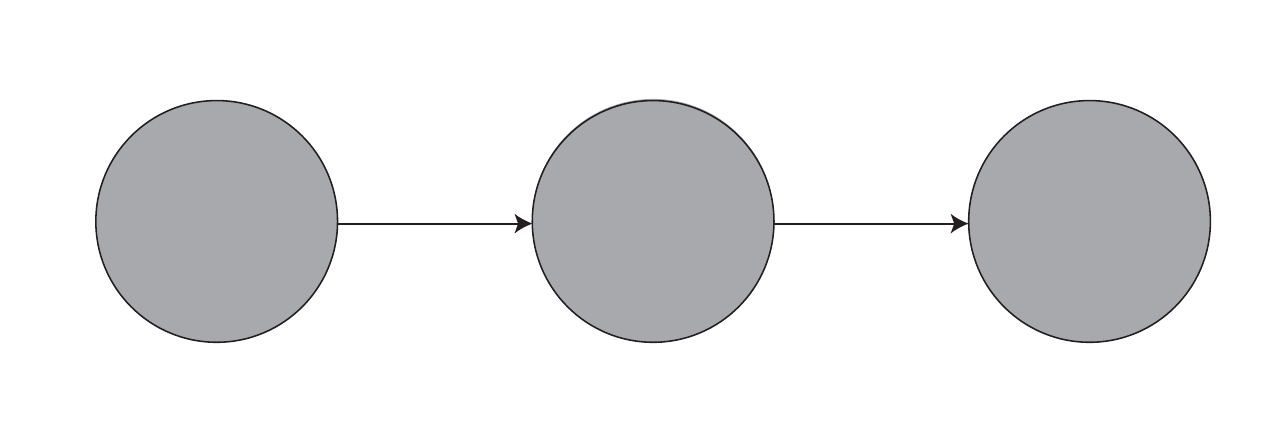}
\put(-360,85){3-dim}
\put(-383,72){Einstein metric}
\put(-360,59){$({\mathcal{H}},{\bmq})$}
\put(-218,85){1+2}
\put(-229,72){Minkowski}
\put(-222,59){$(\mathcal{M},\mathring{\bmq})$}
\put(-85,85){Cylinder at}
\put(-95,72){Spatial infinity}
\put(-85,59){$(\mathbb{R}\times\mathbb{S}^2,\bar{\bmq})$}
\put(-285,75){$\varsigma$}
\put(-139,75){$\omega$}
\caption{Schematic summary of the argument showing that the an asymptote $\mathcal{H}$ satisfying Definition \ref{Definition:AMSI} and the $1+2$ Einstein static Universe are conformally related.}
\label{Fig:ConformalFactors2}
\end{figure}

\subsubsection{From the Einstein metric to $1+2$-Minkowski spacetime}
We begin by observing that if $\bmq$ is conformally related to the $1+2$-Minkowski metric $\mathring{\bmq}$ via a rescaling  of the form $\bmq = \varsigma^{2} \mathring{\bmq}$, then the Schouten tensors of $\bmq$ and $\mathring{\bmq}$ are related via 
\begin{equation}
l_{ij} -\mathring{l}_{ij}= - \frac{1}{\varsigma} D_i D_j \varsigma +
\frac{1}{2 \varsigma^2} D_k \varsigma D^k\varsigma q_{ij}. 
\label{Schouten-transformation}
\end{equation}
Now, defining  
\[
\alpha_i \equiv \varsigma^{-1} D_{i}\varsigma,
\]
then we can rewrite equation \eqref{Schouten-transformation} as 
\[
l_{ij} - \mathring{l}_{ij}= - D_i \alpha_j - \alpha_i \alpha_j + \frac{1}{2} \alpha_k \alpha^k q_{ij}.
\]
Given that $\mathring{l}_{ij} =
0$, one finds that 
\begin{equation}
D_{i} \alpha_j = - l_{ij} - \alpha_i \alpha_j + \frac{1}{2} \alpha_k \alpha^k q_{ij}. \label{alpha-PDE1}
\end{equation}
Multiplying by $q^{il}$ and expanding the covariant derivative in local coordinates
$\underline{x}=(x^\alpha)$, one ends up with the expression
\begin{equation}
\partial_\alpha \alpha^\beta= - l_\alpha{}^\beta - \alpha_\alpha \alpha^\beta + \frac{1}{2} \alpha_\gamma \alpha^\gamma \delta_\alpha^\beta - \Gamma_\alpha{}^\beta{}_\gamma{} \alpha^\gamma. \label{alpha-PDE2}
\end{equation}
This is an overdetermined partial differential equation for the
components of the covector $\bmalpha$. To ensure the existence of a solution to this equation, we make use of \emph{Frobenius theorem} ---see
e.g. \cite{Dun10}, Appendix C. Specifically, for a system of partial
differential equations of $N$ dependent variables given in terms of $n$ independent variables
\[
\frac{\partial \alpha^{\mathcal{A}}}{\partial x^{i}} = \psi_{i}^{\mathcal{A}}(\underline{x},{\bmalpha}), \qquad i=1,\dots,n, \qquad \mathcal{A}=1,...,N,
\]
the \emph{necessary and sufficient condition} to find a unique solution
$\alpha^{\mathcal{A}}=\alpha^{\mathcal{A}}(\underline{x})$ is given by
\begin{eqnarray}
&& \frac{\partial \psi_{\alpha}^{\mathcal{A}}}{\partial x^\beta} - \frac{\partial \psi_{\beta}^{\mathcal{A}}}{\partial x^\alpha }+ \sum_{\mathcal{B}} \left( \frac{\partial \psi _{\alpha}^{\mathcal{A}}}{\partial \alpha^{\mathcal{B}}} \psi_{\beta}^{\mathcal{B}} - \frac{\partial \psi _{\beta}^{\mathcal{A}}}{\partial \alpha^{\mathcal{B}}} \psi_{\alpha}^{\mathcal{B}} \right) = 0. \label{integrability-condition}
\end{eqnarray}
In the case of equation \eqref{alpha-PDE2}, one has that $n=3$ and $N=3$. Thus, making the identification $\psi_\alpha^{\mathcal{A}} \mapsto \partial_\alpha \alpha^{\mathcal{A}}$ and after a lengthy calculation, we find that equation \eqref{integrability-condition} can be rewritten as
\begin{eqnarray}
&& r^l{}_{kji} \alpha_l - 2 \alpha_{[i} D_{j]} \alpha_k + 2 \alpha_l{} q_{k[i} D_{j]} \alpha^l =0  \label{integ-cond}
\end{eqnarray} 
with $r^l{}_{kji}$ denoting the components of the Riemann tensor of $\bmq$. Now, in 3 dimensions, the Riemann tensor is completely determined by the Schouten tensor ---more precisely
\[
r^{k}{}_{l j i} = 2 \left( \delta^k{}_{[j} l_{i]l} - q_{l[j} l_{i]}{}^k \right).
\]
Using this expression for the Riemann tensor, together with the condition $l_{ij}= \digamma q_{ij}$, valid for an Einstein space, in equation \eqref{integ-cond}, we find that the integrability condition is automatically satisfied. Accordingly, there exists a solution $\alpha_\alpha$ to equation \eqref{alpha-PDE2}. From equation \eqref{alpha-PDE1}, one can show that $D_{[i} \alpha_{j]}=0$, so $\alpha_j$ is a closed covector. Hence, it is locally exact. Thus, one can guarantee the existence of a conformal factor locally relating the 3-dimensional Lorentzian metric $\bmq$ and the $1+2$-Minkowski metric
$\mathring{\bmq}$.

\subsubsection{From the $1+2$-Minkowski spacetime to Friedrich's cylinder}
The next step is to find the conformal factor relating the $1+2$-Minkowski to the cylinder at spatial infinity. Starting with the $1+2$-Minkowski metric in the form 
\[
\mathring{\bmq} = \mathbf{d}t \otimes \mathbf{d}t - \mathbf{d}r \otimes \mathbf{d}r - r^2 \mathbf{d}\theta \otimes \mathbf{d}\theta,
\] 
we can write this metric in terms of double null coordinates $u \equiv t-r$ and $v\equiv t+r$ as 
\[
\mathring{\bmq} = \frac{1}{2} \left( \mathbf{d}{v} \otimes \mathbf{d}{u} + \mathbf{d}{u} \otimes \mathbf{d}{v} \right) - \frac{1}{4} (v-u)^2 \mathbf{d}{\theta} \otimes \mathbf{d}{\theta}.
\]
If we introduce the two coordinate transformations
\[
u = \tan U, \qquad v=\tan{V},
\]
and
\[
R= V-U, \qquad T=V+U,
\]
the $1+2$-Minkowski metric transforms to 
\begin{eqnarray*}
&& \mathring{\bmq} = \omega^{-2} \left( \mathbf{d}{T} \otimes \mathbf{d}{T} - \mathbf{d}{R} \otimes \mathbf{d}{R} - \sin^2 R \mathbf{d}{\theta} \otimes \mathbf{d}{\theta} \right), \\
&&\phantom{tilde{\bmeta}} = \omega^{-2} \left( \mathbf{d}{T} \otimes \mathbf{d}{T} - \bmsigma \right), \\
&& \phantom{tilde{\bmeta}} =  \omega^{-2} \bar{\bmq},
\end{eqnarray*}
where $\bar{\bmq}$ is the metric of Friedrich's cylinder and $\omega$ is given by 
\[
\omega \equiv 2 \cos U \cos V \equiv \cos T + \cos R.
\]

\subsubsection{Summary}

The discussion of the previous subsections can be summarised as
follows:

\begin{proposition}
The metric $\bmq$ of an asymptote $\mathcal{H}$ satisfying Definition
\ref{Definition:AMSI} is conformally related to the standard metric of
Friedrich's cylinder $\mathbb{R}\times \mathbb{S}^2$.
\end{proposition}

In the following, we make this relation more precise by recasting the neighbourhood of the asymptote $\mathcal{H}$ in terms of the F-gauge
in which the cylinder at spatial infinity is described.


\section{Conformal Gaussian gauge systems in a neighbourhood of an asymptote}
\label{Section:CGAsymptote}
In this section, we provide the main analysis of the article: the construction of a conformal Gaussian system in a (spacetime) neighbourhood of the asymptote $\mathcal{H}$ satisfying Definition \ref{Definition:AMSI}. More precisely, we show that Definition
\ref{Definition:AMSI} provides enough regularity in the conformal geometric fields to run a stability argument to show that the neighbourhood of the asymptote $\mathcal{H}$ can be covered by a non-intersecting congruence of conformal geodesics extending up to null infinity (and beyond). This congruence is used, in turn, to build on top a representation of spatial infinity \emph{\`{a} la} Friedrich.

\medskip
The construction proceeds in various steps and mimics, to some extent,
the analysis of the conformal extension of static and stationary
spacetimes given in \cite{Fri04,AceVal11}. 

\subsection{A compactified version of the asymptote $\mathcal{H}$}

The construction of a conformal Gaussian system is based on the conformal representation of the asymptotic region of the spacetime $(\tilde{\mathcal{M}},\tilde{\bmg})$ given by the metric \eqref{gcloseH}. Moreover, using arguments similar to those used in Section \ref{Section:HyperboloidToCylinder} it can be assumed, without loss of generality, that the metric $\mathring{\bmq}$ is, in fact, the metric of the unit hyperboloid $\bmell$. Accordingly, in the following, we consider a metric of the form 
\begin{equation}
{\bmg} = -\frac{1}{\zeta^2} \mathbf{d}\zeta\otimes
\mathbf{d}\zeta + \big( \bmell+\breve{\bmq}), \qquad
\breve{\bmq}=o(\zeta). 
\label{UnphysicalAshtekarMetric}
\end{equation}
\medskip
In order to construct a conformal Gaussian system in a neighbourhood of the asymptote $\mathcal{H}$, it is convenient to consider a representation in which the time dimension has a compact extension.  There are several ways of doing this; however, for the
present purposes, probably the most convenient approach is to mimic the discussion of the relation between the F-gauge and hyperboloid representations of spatial infinity for Minkowski spacetime given in Section \ref{Subsection:RelatingSpatialInfinityMinkowski}.

\medskip
In the following introduce new coordinates $(\grave{\tau},\grave{\rho},\grave{\theta}^A)$ in the line element \eqref{UnphysicalAshtekarMetric} via the relations
\[
\grave{\rho} = \zeta \cosh \chi, \qquad \grave{\tau} =\tanh\chi.
\] 
This coordinate transformation is formally identical to the one used in Section \ref{Subsection:RelatingSpatialInfinityMinkowski}, but the geometric interpretation does not follow through as the vector field $\bmpartial_\tau$ is no longer tangent to a congruence of conformal geodesics. Observe that $\grave{\tau}=0$ if and only if $\chi=0$. A computation then shows that 
\[
{\bmg} = \varpi^2\bar\bmeta + \breve{\bmq}, \qquad \varpi \equiv \frac{1}{\sqrt{1-\grave{\tau}^2}}.
\] 
This suggests introducing a new conformal metric $\grave{\bmg}$ via
\[
\grave{\bmg} \equiv \varpi^{-2} \bmg,
\]
so that
\begin{equation}
\grave{\bmg} = \bar\bmeta + \varpi^{-2} \breve{\bmq},
\label{bmggrave}
\end{equation}
where $\bar\bmeta$ is the expression for the Minkowski metric in the F-gauge, equation \eqref{Minkowski:Cylinder}, with the replacements $\tau \mapsto \grave{\tau}$ and $\rho\mapsto \grave{\rho}$. Observe that 
\begin{eqnarray*}
&& \varpi^{-2} \breve{\bmq} = (1-\grave{\tau}^2) \,\breve{\bmq}, \\
&& \phantom{\varpi^{-2} \breve{\bmq}} = \sech^2 \chi \,\breve{\bmq}.
\end{eqnarray*}
In view of the above expression, we make further assumptions independently of Definition \ref{Definition:AMSI}:
\begin{assumption}
\label{Assumption:Extra}
{\em
\begin{itemize}
\item[(i)] The field 
\[
(1-\grave{\tau}^2)\, \breve{\bmq} \qquad (= \sech^2\chi \, \breve{\bmq})
\]
has a suitably regular limit as $\grave{\tau} \rightarrow \pm 1$ (i.e. $\chi \rightarrow
\pm \infty$). 

\item[(ii)] Moreover, it is required that 
\[
(1-\grave{\tau}^2)
\breve{\bmq}(\mathbf{d}\grave{\tau},\mathbf{d}\grave{\tau})
\rightarrow 0, \qquad \mbox{as} \qquad \grave{\tau}\rightarrow \pm 1.
\]
\end{itemize}
 }
\end{assumption}

\begin{remark}
{\em The above assumption is, in fact, a statement about the regularity of the conformal metric in the sets where spatial infinity meets null infinity. A programme to analyse this issue has been started in \cite{Fri98a} ---see also \cite{CFEBook}, Chapter 20 for further discussions on the subject.}
\end{remark}

\begin{remark}
{\em Condition (ii) above ensures that $\grave{\bmg}(\mathbf{d}\grave{\tau},\mathbf{d}\grave{\tau})=0$ for $\tau=\pm 1$, so that the conformal boundary is null as it is to be expected. \emph{Suitably regular in the present context means that the limit is assumed to be sufficiently regular for the subsequent argument to hold.} The analysis of how stringent these conditions are or how this can be encoded goes beyond the scope of this article.} 
\end{remark}

\subsection{Regularisation of the conformal fields}
The metric $\grave{\bmg}$ shares with $\bmg$ the property of being singular at the asymptote $\mathcal{H}$. In order to deal with this difficulty, we follow  the spirit of the analysis of \cite{Fri04}, Section 6, and consider a frame
$\{ \bmc_\bma \}$, $\bma =\bmzero, \, \bmone,\, +,\, -$, with
\[
\bmc_\bmzero \equiv \bmpartial_{\grave{\tau}}, \qquad \bmc_\bmone \equiv 
\grave{\rho} \bmpartial_{\grave{\rho}}, \qquad \bmc_\bmA
\equiv \bmpartial_A = (\bmpartial_+, \bmpartial_-),
\]
where $(\bmpartial_+,\bmpartial_-)$ is a complex null frame on $\mathbb{S}^2$. The associated coframe $\{ \bmalpha^\bma \}$ is given
by
\[
\bmalpha^\bmzero \equiv \mathbf{d}\grave{\tau}, \qquad \bmalpha^\bmone \equiv
\frac{1}{\grave{\rho}}\mathbf{d}\grave{\rho}, \qquad \bmalpha^\bmA =\mathbf{d}\grave{\theta}^A =(\bmomega^+,\bmomega^-).
\]
In terms of this frame, the standard round metric of $\mathbb{S}^2$ is given by
$$
\bmsigma = 2 \left( \bmomega^+ \otimes \bmomega^- + \bmomega^- \otimes \bmomega^+ \right).
$$
Clearly, $\langle \bmalpha^\bma, \bmc_\bmb\rangle = \delta_\bmb{}^\bma$. However, notice that $\{ \bmc_\bma \}$ is not $\grave\bmg$-orthogonal. In the following, we ignore the complications arising from the fact that there is no globally defined basis over
$\mathbb{S}^2$ as there are standard methods to deal with them ---see e.g. \cite{Fri03b,AceVal11,GasVal20}. The components $\grave{g}_{\bma\bmb}\equiv
\grave{\bmg}(\bmc_\bma,\bmc_\bmb)$ of $\grave{\bmg}$ with respect to the frame $\{ \bmc_\bma \}$ are given by 
\[
\grave{g}_{\bma\bmb} = \bar{\eta}_{\bma\bmb} + (1-\grave{\tau}^2) \breve{q}_{\bma\bmb},
\]
where 
\[
( \bar{\eta}_{\bma\bmb} ) =\left(
\begin{array}{cccc}
1 & \grave{\tau} & 0 & 0 \\
\grave{\tau} & -(1-\grave{\tau}{}^2) & 0 & 0 \\
0& 0 & 0 & -2 \\
0 & 0 & -2 & 0
\end{array}
\right), \qquad
\breve{q}_{\bma\bmb} = o(\grave{\rho}).
\]
Similarly, defining $\grave{g}^{\bma\bmb}\equiv
\grave{\bmg}{}^\sharp(\bmalpha^\bma,\bmalpha^\bmb)$ one finds that
\[
\grave{g}^{\bma\bmb} = \bar{\eta}^{\bma\bmb}+Q^{\bma\bmb},
\]
where 
\[
( \bar{\eta}^{\bma\bmb} ) =\left(
\begin{array}{cccc}
(1-\grave{\tau}{}^2) & \grave{\tau} & 0 & 0 \\
\grave{\tau} & -1 & 0 & 0 \\
0& 0 & 0 & -\tfrac{1}{2} \\
0 & 0 &- \tfrac{1}{2} & 0
\end{array}
\right),
\qquad Q^{\bma\bmb}=o(\grave{\rho}^3).
\]

\medskip
In the following, let $\grave{\Gamma}_\bma{}^\bmb{}_\bmc \equiv \langle \bmomega^\bmb ,\grave{\nabla}_\bma \bmc_\bmc \rangle$ denote the connection coefficients of the Levi-Civita connection of the metric $\grave{\bmg}$ with respect to the frame $\{ \bmc_\bma \}$, and $\grave{L}_{\bma\bmb}$ the components of the associated Schouten
tensor. A detailed computation leads to the following:

\begin{lemma}
\label{Lemma:RegularityAsymptoteDetails}
The fields $\grave{g}_{\bma\bmb}$, $\grave{g}{}^{\bma\bmb}$, $\grave{\Gamma}_\bma{}^\bmb{}_\bmc$ and
$\grave{L}_{\bma\bmb}$ extend smoothly to $\mathcal{H}$.  
\end{lemma}

The details of the proof of the above result can be found in Appendix
\ref{Appendix:RegularityAsymptoteDetails}. 

\subsection{Analysis of the $\grave{\bmg}$-conformal geodesic equations}
In the following, we consider the conformal geodesic equations for
the metric $\grave{\bmg}$ in terms of the basis $\{ \bmc_\bma\}$ and
study their solution in a neighbourhood of the asymptote $\mathcal{H}$
with the aim of establishing the existence of a conformal Gaussian
system.  

\subsubsection{The equations}
The conformal geodesic for the metric $\grave{\bmg}$ is given by a spacetime curve
$\big(x^\mu(\tau)\big) =\big( \grave{\tau}(\tau), \grave{\rho}(\tau), \grave{\theta}^A(\tau)
\big)$ with tangent vector $\bmv(\tau)$ and a 1-form
$\grave{\bmbeta}(\tau)$ along the curve such that:
\begin{subequations}
\begin{eqnarray}
&& \dot{\bmx} = \bmv, \label{gbarCGEqn1} \\
&& \grave{\nabla}_{\bmv} \bmv = -2 \langle \grave\bmbeta, \bmv \rangle
   \bmv +
   \grave\bmg(\bmv,\bmv)\grave{\bmbeta}^\sharp, \label{gbarCGEqn2}\\
&& \grave{\nabla}_\bmv \grave{\bmbeta} = \langle \grave\bmbeta, \bmv \rangle \grave\bmbeta
   -\frac{1}{2}\grave\bmg (\grave\bmbeta,\grave\bmbeta) \bmv^\flat +
   \grave{L}(\bmv, \cdot), \label{gbarCGEqn3}
\end{eqnarray}
\end{subequations}
where equation \eqref{gbarCGEqn1} gives the definition of the tangent vector. In order to analyse these equations, consider the expansions
\begin{eqnarray*}
&& \bmv = v^\bma \bmc_\bma, \\
&& \grave\bmbeta = \grave{\beta}_\bma \bmalpha^\bma, \\
&& \grave\bmg = \grave{g}_{\bma\bmb} \bmalpha^\bma \otimes \bmalpha^\bmb,
   \qquad \grave{\bmg}^\sharp = \grave{g}{}^{\bma\bmb} \bmc_\bma \otimes
   \bmc_\bmb, \\
&& \grave{\bmL} = \grave{L}_{\bma\bmb} \bmalpha^\bma \otimes \bmalpha^\bmb.
\end{eqnarray*}
Using the above expansions, equation \eqref{gbarCGEqn1} gives the components
\begin{equation}
\frac{\mathrm{d}\grave{\tau}}{\mathrm{d}\tau} =v^\bmzero, \qquad
\frac{\mathrm{d}\grave{\rho}}{\mathrm{d}\tau} =\grave{\rho} v^\bmone, \qquad
\frac{\mathrm{d}\grave{\theta}^A}{\mathrm{d}\tau} = v^A.
\label{gbarCGEqn1Expanded}
\end{equation}
From equations \eqref{gbarCGEqn2} and \eqref{gbarCGEqn3} one gets
\begin{subequations}
\begin{eqnarray}
&& \frac{\mathrm{d}v^\bma}{\mathrm{d}\tau} +
   \grave\Gamma{}_\bmb{}^\bma{}_\bmc v^\bmb v^\bmc = - 2\grave\beta_\bmd
   v^\bmd v^\bma + \grave{g}{}_{\bmb\bmc} v^\bmb v^\bmc
   \grave{g}{}^{\bma\bmd}\grave{\beta}{}_\bmd, \label{gbarCGEqn2Expanded}
  \\
&& \frac{\mathrm{d}\grave\beta{}_\bma}{\mathrm{d}\tau} -
   \grave\Gamma{}_\bmb{}^\bmc{}_\bma v^\bmb \grave\beta{}_\bmc =
   \grave\beta{}_\bmc v^\bmc \grave\beta{}_\bma +
   \frac{1}{2}\grave{g}{}^{\bmc\bmd} \grave\beta{}_\bmc \grave\beta{}_\bmd
   \grave{g}_{\bma\bmb} v^\bmb + \grave{L}{}_{\bmb\bma} v^\bmb. \label{gbarCGEqn3Expanded}
\end{eqnarray}
\end{subequations}

\subsubsection{The initial data}
\label{Subsubsection:CGData}
Following the discussion of the conformal geodesics on Minkowski spacetime, in the sequel, the initial data of the curve $\big(x^\mu(\tau)\big)$ is chosen so that $\dot{\bmx}$ is $\grave{\bmg}$-normalised and orthogonal to the hypersurface
\[
\mathcal{S}_\star \equiv \{  \grave{\tau} =0 \}.
\]
The normal covector to this hypersurface is given by $\mathbf{d}\grave{\tau}$. Then, it follows that the unit normal vector is given by
\begin{eqnarray*}
    &&\grave{\bmn} \equiv \frac{1}{\sqrt{\grave{\bmg}^{\sharp}(\bmd{\grave{\tau}},\bmd{\grave{\tau}})}} \grave{\bmg}^{\sharp}(\bmd{\grave{\tau}},.), \\
   &&  \phantom{\grave{\bmn}} = \frac{1}{\sqrt{(1-\grave{\tau}^2)(1+\bmQ(\bmd{\grave{\tau}},\bmd{\grave{\tau}}))}} \left( (1-\grave{\tau}^2) \bmpartial_{\grave{\tau}} + \grave{\rho} \grave{\tau} \bmpartial_{\grave{\rho}} + \bmQ(\bmd{\grave{\tau}},.) \right),
\end{eqnarray*}
Thus, on $\mathcal{S}_\star$, one has that 
\begin{equation*}
    \grave{\bmn}_{\star} =
    \frac{1}{\sqrt{1+\bmQ_{\star}(\bmd{\grave{\tau}},\bmd{\grave{\tau}})}}\left(
    \bmpartial_{\grave{\tau}} + \bmQ_{\star}(\bmd{\grave{\tau}},.) \right).
\end{equation*}
From the above, one can readily compute
\[
v^\bmzero_\star \equiv \langle \bmalpha^\bmzero_{\star}, \grave{\bmn}_\star
\rangle, \qquad v^\bmone_\star \equiv \langle \bmalpha^\bmone_{\star},
\grave{\bmn}_\star\rangle, \qquad v^A_\star \equiv \langle \bmalpha^A_{\star},
\grave{\bmn}_\star \rangle.
\]
From the above discussion, it follows then that the initial data for $(x^{\mu}(\tau))$ and $\bmv$ takes the form
\begin{subequations}
\begin{eqnarray}
&& ( x^\mu_\star) = \big(0, \grave{\rho}_\star, \grave{\theta}_\star^A \big), \label{gbarCGEqn1Data} \\
&& v^\bmzero_\star = \frac{1}{\sqrt{1 + \bmQ_{\star}(\bmd{\grave{\tau}},\bmd{\grave{\tau}})}}, \label{gbarCGEqn2aData} \\
&&  v^\bmone_\star =\frac{\bmQ_{\star}(\bmd{\grave{\tau}},\bmd{\grave{\rho}})}{\grave{\rho}_{\star} \sqrt{1+\bmQ_{\star}(\bmd{\grave{\tau}},\bmd{\grave{\tau}})}}, \label{gbarCGEqn2bData} \\
&&  v^A_\star =\frac{\bmQ_{\star}(\bmd{\grave{\tau}},\bmd{\grave{\theta}^{\mathcal{A}}})}{\sqrt{1+\bmQ_{\star}(\bmd{\grave{\tau}},\bmd{\grave{\tau}})}}. \label{gbarCGEqn2cData}
\end{eqnarray}
\label{gbarCGEqnData}
\end{subequations}
As $\bmQ_{\star} = o(\grave{\rho}^3)$, it follows then that
\[
v^\bmzero_\star \rightarrow 1, \qquad  v^\bmone_\star,\;  v^A_\star
\rightarrow 0, \qquad \mbox{as} \quad \grave{\rho}_{\star} \rightarrow 0.
\]
To prescribe initial data for the 1-from $\grave{\bmbeta}$, note that the metric $\grave{\bmg}$ is related to $\tilde{\bmg}$ by the conformal factor $\Theta \equiv \grave{\rho} (1-\grave{\tau}^2)$. Then, $\grave{\bmbeta}_{\star}$ can be written in terms of $\Theta_{\star} \equiv \grave{\rho}_{\star}$ as
\begin{eqnarray*}
    && \grave{\bmbeta}_{\star} = \Theta_{\star}^{-1} \bmd{\Theta_{\star}}, \\
    && \phantom{\grave{\bmbeta}_{\star}} = \frac{1}{\grave{\rho}_{\star}} \bmd{\grave{\rho}}.
\end{eqnarray*}
Thus, $\grave{\bmbeta}_{\star}$ is singular at $\grave{\rho}_{\star}=0$. However, the components $\grave{\beta}_{\bma}$ with respect to $\{ \bmalpha^{\bma} \}$ are regular. In particular, one can show that $\grave{\bmbeta}_{\star} = \bmalpha_{\star}^{\bmone}$, i.e.
\begin{equation}
    \grave{\beta}_{\bmzero \star} =0, \qquad \grave{\beta}_{\bmone \star} =1, \qquad \grave{\beta}_{\bmA \star} =0.
    \label{gbarCGEqn3Data}
\end{equation}
Accordingly, the data described by conditions \eqref{gbarCGEqnData} and \eqref{gbarCGEqn3Data} are regular at $\mathcal{H}$. 

\subsubsection{The solution to the conformal geodesic equations in a neighbourhood of $\mathcal{H}$}
In this section, we make use of the information gathered in the previous sections to show the existence in a neighbourhood of the asymptote $\mathcal{H}$ of congruence of non-intersecting conformal geodesics extending smoothly beyond null infinity. This congruence gives rise, in a natural way, to a conformal Gaussian gauge system. The argument in this section makes use of the theory of perturbations of ordinary differential equations.

\medskip
The first step in the analysis is the observation that the system \eqref{gbarCGEqn1Expanded} and
\eqref{gbarCGEqn2Expanded}-\eqref{gbarCGEqn3Expanded} with initial data given by \eqref{gbarCGEqnData} and \eqref{gbarCGEqn3Data} can be solved exactly on $\mathcal{H}$. 

\begin{lemma}
\label{Lemma:CGAsymptote}
The unique solution to the conformal geodesic equations \eqref{gbarCGEqn1Expanded} and
\eqref{gbarCGEqn2Expanded}-\eqref{gbarCGEqn3Expanded} with initial data given by \eqref{gbarCGEqnData} and \eqref{gbarCGEqn3Data} for $\grave\rho_{\star}=0$ is given by 
\begin{eqnarray*}
&& \big(x^\mu(\tau)\big) = (\tau, 0, \grave{\theta}^A_{\star}), \qquad \grave{\tau}(\tau) = \tau, \quad \grave{\rho}(\tau) = \grave{\rho}_{\star}=0, \quad \grave{\theta}^{\bmA}(\tau) = \grave{\theta}^{\bmA}_{\star}, \\
&& v^\bmzero(\tau) =1, \qquad v^\bmone(\tau)=0, \qquad v^A(\tau) =0,
  \\
&& \grave{\beta}_\bmzero(\tau) =0, \qquad \grave{\beta}_\bmone(\tau)
   =1, \qquad \grave{\beta}_A(\tau) = 0. 
\end{eqnarray*}
This solution is smooth for all $\tau \in \mathbb{R}$. In particular,
it extends smoothly beyond the interval $[-1,1]$. 
\end{lemma}

\begin{proof}
The proof of the lemma follows from the observation that on $\mathcal{H}$ (i.e. $\grave{\rho} =0$), the metric $\grave{\bmg}$ coincides with the metric $\bar{\bmeta}$. As discussed in Section \ref{Section:MinkowskiCylinder}, Lemma \ref{Lemma:CGMinkowski}, the solution to the conformal geodesic equations with the given data is given by
\[
\big( x^\mu(\tau) \big)=(\tau,0,\grave{\theta}^A_\star), \qquad
\dot{\bmx} = \bmpartial_{\grave{\tau}}, \qquad \grave{\bmbeta}= \frac{\mathbf{d}\grave{\rho}}{\grave{\rho}_\star}.
\]
Contracting with the frame $\{ \bmc_\bma \}$ and coframe $\{ \bmalpha^\bma \}$, one obtains the result.
\end{proof}
From the above result, making use of the regularity of the fields appearing in the conformal geodesic equations \eqref{gbarCGEqn1Expanded} and \eqref{gbarCGEqn2Expanded}-\eqref{gbarCGEqn3Expanded} and the initial conditions \eqref{gbarCGEqnData} and \eqref{gbarCGEqn3Data} in a neighbourhood of $\mathcal{H}$, one obtains the following:

\begin{lemma}
\label{Lemma:StabilityCG}
There exists $\grave{\rho}_\bullet>0$ such that if $\grave\rho_\star\in[0,\grave\rho_\bullet)$ then the system \eqref{gbarCGEqn1Expanded} and \eqref{gbarCGEqn2Expanded}-\eqref{gbarCGEqn3Expanded} with initial conditions \eqref{gbarCGEqnData} and \eqref{gbarCGEqn3Data} has a unique smooth solution with existence interval extending beyond $[-1,1]$ ---e.g. $\tau\in [-\tfrac{3}{2},\tfrac{3}{2}]$. 
\end{lemma}

\begin{proof}
The result follows from the stability theory of ordinary differential
equations. In particular, the regularity of the components
$\grave{g}_{\bma\bmb}$, $\grave{g}^{\bma\bmb}$ and
$\grave{L}_{\bma\bmb}$ in a neighbourhood of the asymptote
$\mathcal{H}$ and up to and beyond the sets given by the conditions
$\tau=\pm 1$ allows making use of Theorem \ref{Theorem:ODEStability} in
Appendix \ref{Appendix:ODETheory}.
\end{proof}

\begin{remark}
{\em By varying the starting point $p\in\mathcal{S}_\star$ subject to
  the condition $\grave\rho_\star(p)\in[0,\grave\rho_\bullet)$ ensures
that a neighbourhood of $\mathcal{H}$ can be covered by conformal
geodesics. In the next subsection, it will be shown that, possibly by
reducing $\grave\rho_\bullet$ this
congruence is non-intersecting.}
\end{remark}

\begin{remark}
{\em The curves of the congruence extend beyond $\mathscr{I}^\pm$ as their existence interval extend beyond $ [-1,1]$. }
\end{remark}

\subsection{The construction of the conformal Gaussian system}
In this section, we discuss how the congruence of conformal geodesics obtained in Lemma~\ref{Lemma:CGAsymptote} implies the existence of a conformal Gaussian system
in a neighbourhood of the asymptote $\mathcal{H}$. 

\subsubsection{General considerations}
Given a point $p\in \mathcal{S}_\star$ with coordinates $\underline{x}(p)=(\grave{\rho}(p),\grave\theta{}^A(p))$ such that $\grave{\rho}(p)<\grave{\rho}_\bullet$, these coordinates are propagated off $\mathcal{S}_\star$ by requiring them to be constant along the \emph{unique conformal geodesic} with data as given in Lemma \ref{Lemma:StabilityCG} passing through $p$. In order to differentiate between the propagated coordinates and those on $\mathcal{S}_\star$, we denote the former by $(\rho,\theta^A)$. Points along the conformal geodesic passing through $p$ are labelled using the parameter $\tau$ of the curve. In this way, by varying the point $p$ and as long as the congruence of conformal geodesics is non-intersecting, one obtains \emph{conformal Gaussian coordinates} $\bar{x}=(\tau,\rho,\theta^A)$. 

\begin{remark}
{\em It should be stressed that given a point $q$ in the neighbourhood of $\mathcal{H}$ covered by the curves of Lemma \ref{Lemma:StabilityCG} and described, respectively, by coordinates $\big(\grave\tau(q),\grave\rho(q),\grave\theta^A(q)\big)$ and $\big( \tau(q),\rho(q),\theta^A(q) \big)$, one has, in general, that $\grave\rho(q)\neq \rho(q)$, $\grave\theta^A(q)\neq \theta^A(q)$. }
\end{remark}

\subsubsection{Analysis of the Jacobian}
In order to ensure that the collection $(\tau,\rho,\theta^A)$ gives rise to a well-defined coordinate system in a neighbourhood of $\mathcal{H}$, we need to consider the Jacobian determinant of the change of coordinates
\[
(\tau,\rho,\theta^A)\mapsto (\grave\tau,\grave\rho,\grave\theta^A).
\]
In order to ease the presentation in the following, we restrict the discussion to the transformation between non-angular coordinates $(\grave\tau,\grave\rho)\mapsto (\tau,\rho)$. The full analysis follows in a similar manner at the expense of lengthier computations. Writing 
\[
\grave\tau =\grave\tau(\tau,\rho), \qquad \grave\rho=\grave\rho(\tau,\rho),
\] 
the associated Jacobian determinant is given by
\[
\frac{\partial (\grave\tau,\grave\rho)}{\partial(\tau,\rho)} =
\left|
\begin{array}{cc}
\displaystyle\frac{\partial\grave\tau}{\partial\tau} &
                                         \displaystyle \frac{\partial\grave\tau}{\partial\rho}\\
\displaystyle\frac{\partial\grave\rho}{\partial\tau} & \displaystyle\frac{\partial\grave\rho}{\partial\rho} 
\end{array}
\right|.
\]
Now, from Lemma \ref{Lemma:CGAsymptote}, it follows that
\[
\frac{\partial\grave\tau}{\partial\tau}\bigg|_{\mathcal{H}} =1, \qquad \frac{\partial\grave\tau}{\partial\rho}\bigg|_{\mathcal{H}}=0,
\]
so that, in fact, one has
\[
\frac{\partial
  (\grave\tau,\grave\rho)}{\partial(\tau,\rho)}\bigg|_{\mathcal{H}} =
\frac{\partial\grave\rho}{\partial\rho}\bigg|_{\mathcal{H}}. 
\]
Differentiating the second equation in \eqref{gbarCGEqn1Expanded} one
obtains 
\[
\frac{\mathrm{d}}{\mathrm{d}\tau}\left(
  \frac{\partial\grave\rho}{\partial\rho}\right) = v^\bmone
\frac{\partial\grave\rho}{\partial\rho} +\grave{\rho} \frac{\partial
  v^\bmone}{\partial \rho},
\]
from where it follows 
\[
\frac{\mathrm{d}}{\mathrm{d}\tau}\left(
  \frac{\partial\grave\rho}{\partial\rho}\right)\bigg|_{\mathcal{H}} =0.
\]
In other words, the partial derivative $\partial\grave\rho/\partial\rho$ is constant  on $\mathcal{H}$. To evaluate the constant, it is observed that by construction $\grave\rho=\rho$ on $\mathcal{S}_\star$. Accordingly, one concludes that
\[
 \frac{\partial\grave\rho}{\partial\rho} \bigg|_{\mathcal{H}}=1.
\]
It then follows from the above argument that 
\[
\frac{\partial (\grave\tau,\grave\rho)}{\partial(\tau,\rho)}\bigg|_{\mathcal{H}} =1.
\]
By continuity, it follows that the Jacobian determinant is non-zero in a neighbourhood of $\mathcal{H}$. In order to ensure that this condition holds, it may be necessary to reduce the constant $\grave\rho_\bullet$ in Lemma~\ref{Lemma:StabilityCG}. 
\begin{remark}
{\em It follows from the fact that the Jacobian of the transformation between the coordinates $(\grave\tau,\grave\rho,\grave\theta^A)$ and $(\tau,\rho,\theta^A)$ is non-zero that the congruence of conformal geodesics in a neighbourhood of $\mathcal{H}$ is non-intersecting as the intersection of curves would indicate a singularity in the coordinate system given by $(\tau,\rho,\theta^A)$. }
\end{remark}

\subsubsection{The conformal factor associated with the congruence}
The existence of a congruence of conformal geodesics in a neighbourhood of the asymptote $\mathcal{H}$ ensures that the family of \emph{conformal factors} given by Proposition \ref{canconical-confromalFactor} is well defined. This, in turn, completes the construction of a conformal Gaussian system in a neighbourhood of $\mathcal{H}$. 
\medskip
Consistent with the discussion of the initial data for the congruence
of conformal geodesics in Subsection \ref{Subsubsection:CGData}, one has that
\[
\Theta_\star =\Omega_\star =\rho.
\]
Moreover, from \eqref{gbarCGEqnData} and \eqref{gbarCGEqn3Data}, it follows that 
\[
\dot{\Theta}_{\star} = \rho \langle \grave{\bmbeta}_{\star}, \bmv_{\star} \rangle = \rho v^{\bmone}_{\star} = o(\rho^3).
\]
For $\ddot{\Theta}_{\star}$, let $\grave{\bmh}^{\sharp}$ denote the intrinsic contravariant metric on $\mathcal{S}_{\star}$, written explicitly as
\begin{equation*}
    \grave{\bmh}^{\sharp} = - \rho^2 \bmpartial_{\rho} \otimes \bmpartial_{\rho} - \bmsigma^{\sharp} + \bmQ_{\star}.
\end{equation*}
Then, one has
\begin{equation*}
    \ddot{\Theta}_{\star} = 2 \rho \grave{\bmh}^{\sharp}(\grave{\bmbeta}_{\star},\grave{\bmbeta}_{\star}).
\end{equation*}
The initial data for $\grave{\bmbeta}$ implies
\begin{equation*}
    \grave{\bmh}^{\sharp}(\grave{\bmbeta}_{\star},\grave{\bmbeta}_{\star}) = - \left( 1- \frac{1}{\rho^2} \bmQ_{\star}(\bmd{\rho},\bmd{\rho}) \right).
\end{equation*}
Hence, 
\begin{equation*}
    \ddot{\Theta}_{\star} = -2 \rho  \left( 1- \frac{1}{\rho^2} \bmQ_{\star}(\bmd{\rho},\bmd{\rho}) \right).
\end{equation*}
Observe that 
\begin{equation*}
    \lim_{\rho \to 0}\frac{\bmQ_{\star}(\bmd{\rho},\bmd{\rho})}{\rho} =0.
\end{equation*}
Then, the canonical conformal factor $\Theta$, defined in the neighbourhood of $\mathcal{H}$ where the congruence is non-intersecting, is given by
\begin{eqnarray*}
    && \Theta = \Theta_{\star} + \dot{\Theta}_{\star} \tau + \frac{1}{2} \ddot{\Theta}_{\star} \tau^2, \\
    && \phantom{\Theta} = \rho \left( 1+ v^{\bmone}_{\star} \tau - \left( 1-\frac{1}{\rho^2} \bmQ_{\star}(\bmd{\rho},\bmd{\rho}) \right) \tau^2 \right), 
\end{eqnarray*}
The conformal boundary at $\Theta=0$ can be identified by the conditions
\begin{equation}
    \rho=0, \quad \text{or } \quad 1+ v^{\bmone}_{\star} \tau - \left( 1-\frac{1}{\rho^2} \bmQ_{\star}(\bmd{\rho},\bmd{\rho}) \right) \tau^2 =0.
\end{equation}
Solving for $\tau$ gives 
\begin{equation*}
    \tau = \frac{v^{\bmone}_{\star} \pm \sqrt{(v^{\bmone }_{\star})^{2} + 4 \left( 1-(1/\rho^2) \bmQ_{\star}(\bmd{\rho},\bmd{\rho}) \right)}}{2 \left( 1- (1/\rho^2) \bmQ_{\star}(\bmd{\rho},\bmd{\rho}) \right)}.
\end{equation*}
It can be shown that 
\begin{equation*}
    \tau = \pm 1 + O(\rho), 
\end{equation*}
so that $\tau \to \pm 1$ as $\rho \to 0$. Now, let
\begin{eqnarray*}
    && \tau_{-} = \frac{v^{\bmone}_{\star} - \sqrt{(v^{\bmone }_{\star})^{2} + 4 \left( 1-(1/\rho^2) \bmQ_{\star}(\bmd{\rho},\bmd{\rho}) \right)}}{2 \left( 1- (1/\rho^2) \bmQ_{\star}(\bmd{\rho},\bmd{\rho}) \right)}, \\
    && \tau_{+} = \frac{v^{\bmone}_{\star} + \sqrt{(v^{\bmone }_{\star})^{2} + 4 \left( 1-(1/\rho^2) \bmQ_{\star}(\bmd{\rho},\bmd{\rho}) \right)}}{2 \left( 1- (1/\rho^2) \bmQ_{\star}(\bmd{\rho},\bmd{\rho}) \right)}.
\end{eqnarray*}
Thus, one defines, in analogy to the case of Minkowski spacetime
in the F-gauge (see Section~\ref{Section:MinkowskiCylinder}) the set
\begin{equation*}
    \mathcal{M}_{\rho_{\bullet}} \equiv \{ p \in \mathbb{R}^{4} | \tau_{-} \leq \tau(p) \leq \tau_{+}, 0 \leq \rho(p) < \rho_{\bullet}  \},
\end{equation*}
and 
\begin{eqnarray*}
    && I \equiv \{ p \in \mathcal{M}_{\rho_{\bullet}} | |\tau(p)| <1, \rho(p) =0 \}, \\
    && I^{\pm} \equiv \{ p \in \mathcal{M}_{\rho_{\bullet}} | \tau(p) = \pm 1, \rho(p) =0 \}.
\end{eqnarray*}
Finally, future and past null infinities $\mathscr{I}^{\pm}$ can be defined as
\begin{equation*}
    \mathscr{I}^{\pm} \equiv \{ p \in \mathcal{M}_{\rho_{\bullet}} | \tau(p) = \tau_{\pm}, | 0 < \rho(p) < \rho_{\bullet} \}.
\end{equation*}
\begin{remark}
{\em The set $I$ (the cylinder at spatial infinity) coincides with the asymptote $\mathcal{H}$.}
\end{remark}
\begin{remark}
{\em Observe that in contrast to the representation of Minkowski in Section \ref{Section:MinkowskiCylinder}, the representation obtained in this section \emph{is not horizontal} in the sense that the location of null infinity is not given by the condition $\tau=\pm 1$. A horizontal representation can be obtained by considering a more general initial conformal factor of the form $\Theta_\star = \varkappa\rho $, $\varkappa|_{\rho=0}=1$,  with $\varkappa$ suitably chosen. }
\end{remark}


\subsubsection{Main statement}
We summarise the discussion of the previous subsections in the following theorem, which is a more detailed version of the statement presented in the introductory section:

\begin{theorem}
\label{Theorem:MainPreciseFormulation}
Let $(\mathcal{\tilde{M}},\tilde{\bmg})$ denote a spacetime satisfying Definition \ref{Definition:AMSI} of an asymptotically Minkowskian spacetime at spatial infinity with asymptote $\mathcal{H}$ which, in addition, satisfies Assumption \ref{Assumption:Extra}. Then there exists a neighbourhood $\mathcal{M}_{\rho_\bullet}$ of
$\mathcal{H}$ which can be covered by a conformal Gaussian coordinate system $(x^\mu)\equiv (\tau,\rho,\theta^A)$. The domain $\mathcal{M}_{\rho_\bullet}$ includes  portions of future and past null infinity which meet with the asymptote $\mathcal{H}$. In this gauge, the asymptote $\mathcal{H}$ coincides with the $1+2$-dimensional Einstein static cylinder (Friedrich's cylinder). 
\end{theorem}


\appendix

\section{The basic stability theorem for ordinary differential
  equations depending on a parameter}
\label{Appendix:ODETheory}

In the following let 
\begin{equation}
\mathbf{X}' = \mathbf{F}(t, \mathbf{X}, \lambda), \qquad
\mathbf{X}(0)=\mathbf{X}_\star,
\label{ODESystem}
\end{equation}
denote an initial value problem for a $N$-dimensional vector-value ordinary
differential equation depending on a parameter $\lambda$. One has the
following result ---see \cite{Har87},  Theorem 2.1 in page 94 and
Corollary 4.1 in page 101. 

\begin{theorem}
\label{Theorem:ODEStability}
Let $\mathbf{F}(t, \mathbf{X}, \lambda)$ be continuous on an open set
$\mathcal{E}\subset
\mathbb{R}\times\mathbb{R}^N\times\mathbb{R}$ consisting of points $(t,
\mathbf{X}, \lambda)$ with the property that
for every $(t_\star, \mathbf{X}_\star, \lambda)\in \mathcal{E}$, the
initial value problem \eqref{ODESystem} with $\lambda$ fixed has a
unique solution $\mathbf{X}(t)=
\mathbf{X}(t,t_\star,\mathbf{X}_\star,\lambda)$. Let
$\omega_-<t<\omega_+$ be the maximal interval of existence of
$\mathbf{X}(t,t_\star,\mathbf{X}_\star,\lambda)$. Then
$\omega_+=\omega_+(t_\star,\mathbf{X}_\star,\lambda)$ (respectively
$\omega_-=\omega_-(t_\star,\mathbf{X}_\star,\lambda)$) is a lower
(respectively upper) semicontinuous function\footnote{The lower
  semicontinuity of $\omega_+$ at
  $(t_\bullet,\mathbf{X}_\bullet,\lambda_\bullet)$ means that if
  $t_\sharp < \omega_+(t_\bullet,\mathbf{X}_\bullet,\lambda_\bullet)$
  then $\omega_+(t,\mathbf{X},\lambda)\geq t_\sharp$ for
  all $(t,\mathbf{X},\lambda)$ near
  $(t_\bullet,\mathbf{X}_\bullet,\lambda_\bullet)$. Upper
semicontinuity is similarly defined.} of
$(t_\star,\mathbf{X}_\star,\lambda)\in \mathcal{E}$ and
$\mathbf{X}(t,t_\star,\mathbf{X}_\star,\lambda)$ is continuous on the
set 
\[
\{  \omega_-<t<\omega_+, \; (t_\star, \mathbf{X}_\star,\lambda)\in \mathcal{E} \}.
\]
Moreover, if $\mathbf{F}(t, \mathbf{X}, \lambda)$ is of class $C^m$,
$m\geq 1$ on $\mathcal{E}$, then
$\mathbf{X}(t,t_\star,\mathbf{X}_\star,\lambda)$ is of class $C^m$ on
its domain of existence.
\end{theorem}

So, it follows from the above theorem that if $\mathbf{F}$ satisfies
the conditions of the theorem above, and the problem
\[
\mathbf{X}' = \mathbf{F}(t, \mathbf{X},0), \qquad
\mathbf{X}(0)=\mathbf{X}_\star,
\]
has a solution with existence interval $t\in(\omega_-,\omega_+)$, then
any other solution to \eqref{ODESystem} with $\lambda$ sufficiently
close to $0$ will have, at least the same existence interval. In the
above statement the functions $\omega_+$ and $\omega_-$ can take the
values $\infty$ and $-\infty$. 

\section{Details of the proof of Lemma \ref{Lemma:RegularityAsymptoteDetails}}
\label{Appendix:RegularityAsymptoteDetails}
In this section, we show that the fields $\grave{\Gamma}_\bma{}^\bmb{}_\bmc$ and $\grave{L}_{\bma\bmb}$ extend smoothly to $\mathcal{H}$.  

\subsubsection{Connection coefficients}
The connection coefficients of the Levi-Civita connection $\grave\nabla$ of the metric $\grave\bmg$ with respect to the frame $\{ \bmc_\bma \}$
are defined by the relation
\[
\grave\nabla_\bma \bmc_\bmb = \grave\Gamma_\bma{}^\bmc{}_\bmb \bmc_\bmc.
\]
It can be readily verified that $[\bmc_\bma,\bmc_\bmb]=0$. Thus, it
follows from the Kulkarni formula that
\begin{equation}
\grave\Gamma_\bma{}^\bmb{}_\bmc =\frac{1}{2}\grave{g}^{\bmb\bmd} \big(
\bmc_\bmc(\grave{g}_{\bma\bmd})+\bmc_\bma(\grave{g}_{\bmd\bmc})
-\bmc_\bmb(\grave{g}_{\bma\bmc})\big).
\label{KulkarniFormula}
\end{equation}

\medskip
In the following, it will be shown that the connection coefficients $\grave\Gamma_\bma{}^\bmb{}_\bmc$ are regular at $\grave{\rho}=0$. The analysis of the connection coefficients requires further specification of the coordinates. In the following, we will make use of Gaussian coordinates of the asymptote $\mathcal{H}$ ---that is, the coordinates $(\grave{\tau},\grave{\theta}^A)$ on $\mathcal{H}$ are propagated on the bulk of the spacetime using the relation: 
\begin{equation}
\grave{\nabla}_{\bmpartial_{\grave{\rho}}}\bmpartial_\bmi =0, \qquad
\bmi=\bmzero \, \text{or } \bmA,
\label{GaussianCoordinates}
\end{equation}
where $\bmpartial_{\bmi} \equiv \{ \bmpartial_{\grave{\tau}}, \bmpartial_{\grave{\theta}^A} \}$. Using this condition, one can readily verify that 
\[
\grave{\Gamma}_\bmone{}^\bma{}_\bmi = 0.
\]
The other components of the connection coefficients are not determined by the gauge condition \eqref{GaussianCoordinates} but can be directly computed using Kulkarni's formula \eqref{KulkarniFormula}. One finds that 
\begin{eqnarray*}
&& \grave{\Gamma}_\bmzero{}^\bmzero{}_\bmzero{} = O(\grave{\rho}) \\
&& \grave{\Gamma}_\bmzero{}^\bmone{}_\bmzero{} = -1 + O(\grave{\rho}^2) \\
&& \grave{\Gamma}_\bmzero{}^\bmA{}_\bmzero{} = o(\grave{\rho}^2) \\
&& \grave{\Gamma}_\bmone{}^\bmzero{}_\bmone{} = -  \grave{\tau} (1-\grave{\tau}^2) + O(\grave{\rho})\\
&& \grave{\Gamma}_\bmone{}^\bmone{}_\bmone{}= - \grave{\tau}^2 +O(\grave{\rho}) \\
&& \grave{\Gamma}_\bmone{}^\bmA{}_\bmone{} = O(\grave{\rho}^2) \\
&& \grave{\Gamma}_\bmzero{}^\bmzero{}_\bmA{} = O(\grave{\rho}^3) \\
&& \grave{\Gamma}_\bmzero{}^\bmone{}_\bmA{} = O(\grave{\rho}^2) \\
&& \grave{\Gamma}_\bmA{}^\bmone{}_\bmB{} = O(\grave{\rho}) \\
&& \grave{\Gamma}_\bmzero{}^\bmA{}_\bmB{} = O(\grave{\rho}) \\
&& \grave{\Gamma}_\bmA{}^\bmzero{}_\bmB{} =  O(\grave{\rho}) \\
&& \grave{\Gamma}_\bmA{}^\bmA{}_\bmB{} =  o(\grave{\rho}^3)\\
&& \grave{\Gamma}_\bmB{}^\bmA{}_\bmB{} = o(\grave{\rho}^4) \\
&& \grave{\Gamma}_\bmA{}^\bmA{}_\bmA{} = o(\grave{\rho}^4) \\
\end{eqnarray*}

Thus, the connection coefficients $\grave{\Gamma}{}_\bma{}^\bmb{}_\bmc$
are regular with respect to the coordinates $(\grave{\tau},\grave{\rho})$ in a
neighbourhood of $\mathcal{H}$. 
\subsubsection{The components of the Schouten tensor}
The Schouten tensors $\tilde{L}_{ab}$ and $\grave{L}_{ab}$ are related by the formula
\[
\grave{L}_{ab} - \tilde{L}_{ab} = -\frac{1}{\Theta}\grave{\nabla}_a
\grave{\nabla}_b\Theta + \frac{1}{2\Theta^2}\grave{\nabla}_c\Theta
\grave{\nabla}^c\Theta \grave{g}_{ab}, 
\]
where $\Theta \equiv \grave{\rho} (1-\grave{\tau}^2)$ is the conformal factor relating $\tilde{\bmg}$ and $\grave{\bmg}$. In view of the Einstein field equations, we have that
$\tilde{L}_{ab}=0$ so that 
\begin{equation}
\grave{L}_{ab} =  -\frac{1}{\Theta}\grave{\nabla}_a
\grave{\nabla}_b \Theta+ \frac{1}{2\Theta^2}\grave{\nabla}_c\Theta
\grave{\nabla}^c\Theta \grave{g}_{ab}. 
\label{UnphysicalSchouten}
\end{equation}

\begin{remark}
{\em The previous expression is singular at $\grave{\rho}=0$. We need to compute the components $\grave{L}_{\bma\bmb}$ with respect to $\{ \bmc_\bma \}$.}
\end{remark}

A direct computation shows that
\begin{eqnarray*}
&& \grave{\nabla}_{\bma} \Theta = \bmc_{\bma} (\Theta)  = \begin{cases}
-2 \grave{\tau} \grave{\rho} & \text{if $\bma = \bmzero$}, \\
\grave{\rho} (1-\grave{\tau}^2) & \text{if $\bma = \bmone$}, \\
0 & \text{if $\bma \neq \bmzero, \bmone$},
\end{cases}\\
&& \grave{\nabla}_\bma\grave{\nabla}_\bmb \Theta = \bmc _{\bma} (\bmc_{\bmb} (\Theta)) - \grave{\Gamma}_\bma{}^\bmc{}_\bmb{} \grave{\nabla}_{\bmc} \Theta,
\end{eqnarray*}
where $\bmc_{\bma} (\bmc_{\bmb} (\Theta))$ is given by
\begin{eqnarray*}
\bmc_{\bma} (\bmc_{\bmb} (\Theta)) = \begin{cases}
-2 \grave{\tau} \grave{\rho} & \text{if $\bma \neq \bmb$ and $\bma = \bmzero, \bmone$}, \\
-2 \grave{\rho} & \text{if $\bma = \bmb = \bmzero$}, \\
\grave{\rho} (1-\grave{\tau}^2) & \text{if $\bma = \bmb = \bmone$}, \\
0 & \text{if $\bma$, $\bmb \neq \bmzero , \bmone$},
\end{cases}
\end{eqnarray*}
Thus, the first term on the right-hand side of \eqref{UnphysicalSchouten} is regular. Moreover,
\[
\grave{\nabla}_\bmc\Theta \grave{\nabla}^\bmc \Theta = - \grave{\rho}^2 (1-\grave{\tau}^2)^2,
\]
so that the second term in \eqref{UnphysicalSchouten} is regular. Hence, the components $\grave{L}_{\bma\bmb}$ are regular with respect to the coordinates $(\grave{\tau}, \grave{\rho})$ in a neighbourhood of $\mathcal{H}$. In conclusion, the components of $\grave{L}_{\bma \bmb}$ coincide, on $\mathcal{H}$, with those of the Schouten tensor of Minkowski spacetime in the F-gauge. 



\section{Conformal geodesics on Minkowski spacetime}
\label{Appendix:CGMinkowski}

In this section, we show that the vector $\bmpartial_\tau$ in the
F-gauge representation of the Minkowski spacetime of Section
\ref{Section:MinkowskiCylinder} is a conformal geodesic ---see
Lemma \ref{Lemma:CGMinkowski}. Thus, it should satisfy the equations 
\begin{subequations}
\begin{eqnarray}
&& \bar\nabla_{\bmpartial_\tau} \bmpartial_\tau =
   -2\langle\bar\bmbeta,\bmpartial_\tau\rangle \bmpartial_\tau +
   \bar\bmeta(\bmpartial_\tau,\bmpartial_\tau)\bar\bmbeta^\sharp, \label{MinkowskiCG1}\\
&& \bar\nabla_{\bmpartial_\tau} \bar\bmbeta= \langle \bar\bmbeta,\bmpartial_\tau\rangle
   \bar\bmbeta - \frac{1}{2}\bar\bmeta^\sharp(\bar\bmbeta,\bar\bmbeta)
   (\bmpartial_\tau)^\flat + \bar{\mathbf{L}}(\bmpartial_\tau, \cdot ). \label{MinkowskiCG2}
\end{eqnarray}
\end{subequations}
In particular, one has that
\[
\bar\bmeta(\bmpartial_\tau,\bmpartial_\tau)=1.
\]
In the following, we make use of the expansion
\[
\bar\bmbeta = \bar\beta_0 \mathbf{d}\tau +\bar\beta_1 \mathbf{d}\rho.
\]
A computation gives that
\[
\bar{\mathbf{L}} = \frac{1}{2}\mathbf{d}\tau\otimes\mathbf{d}\tau +
\frac{\tau}{2\rho}(\mathbf{d}\tau\otimes \mathbf{d}\rho +
\mathbf{d}\rho\otimes \mathbf{d}\tau) +
\frac{\tau^2-1}{2\rho^2}\mathbf{d}\rho\otimes\mathbf{d}\rho +\frac{1}{2}\bmsigma.
\]
So, in particular,
\[
\bar{\mathbf{L}}(\bmpartial_\tau,\cdot) = \frac{1}{2}\mathbf{d}\tau + \frac{\tau}{2\rho}\mathbf{d}\rho.
\]
One also has that 
\begin{eqnarray*}
&& \bar\nabla_{\bmpartial_\tau} \mathbf{d}\tau = -\bar\Gamma_0{}^0{}_0
   \mathbf{d}\tau - \bar\Gamma_0{}^0{}_1\mathbf{d}\rho, \\
&& \bar\nabla_{\bmpartial_{\tau}} \bmpartial_{\tau} = \bar\Gamma_{0}{}^{0}{}_{0} \bmpartial_{\tau} + \bar\Gamma_{0}{}^{1}{}_{0} \bmpartial_{\rho}, \\
&& \bar\nabla_{\bmpartial_\tau} \mathbf{d}\rho =
   -\bar\Gamma_0{}^1{}_0\mathbf{d}\tau -\bar\Gamma_0{}^1{}_1\mathbf{d}\rho,
\end{eqnarray*}
so that 
\[
\bar\nabla_{\bmpartial_\tau} \bar\bmbeta = (\bmpartial_{\tau} \bar\beta_{0} - \bar\beta_{0} \bar{\Gamma}_{0}{}^{0}{}_{0}- \bar\beta_{1} \bar{\Gamma}_{0}{}^{1}{}_{0}) \bmd{\tau} + (\bmpartial_{\tau} \bar\beta_{1} - \bar\beta_{0} \bar{\Gamma}_{0}{}^{0}{}_{1}- \bar\beta_{1} \bar{\Gamma}_{0}{}^{1}{}_{1}) \bmd{\rho},
\]
where
\begin{eqnarray*}
&& \bar\Gamma_0{}^0{}_0 = \tau, \qquad \bar\Gamma_0{}^1{}_0 =-\rho, \\
&& \bar\Gamma_0{}^1{}_1 =- \tau.
\end{eqnarray*}

\medskip
Substituting all of the above in equation \eqref{MinkowskiCG1} one
concludes that
\[
\bar\beta_0 =0, \qquad \bar\beta_1 = \frac{1}{\rho},
\]
so that the vector $\bmpartial_\tau$ is tangent to timelike conformal
geodesics with 
\[
\bar\bmbeta = \frac{1}{\rho}\mathbf{d}\rho.
\]
Equation \eqref{MinkowskiCG2} can be shown to be satisfied identically.

\medskip
Finally, let us find the connection to the physical spacetime. Recalling that 
\[
\bar\bmeta = \Theta^2 \tilde\bmeta, \qquad \Theta=\rho(1-\tau^2),
\]
one has that
\[
\tilde\bmbeta = \bar\bmbeta + \mathbf{d}\ln\Theta,
\]
so that 
\[
\tilde\bmbeta = \frac{2}{\rho}\mathbf{d}\rho -
\frac{2\tau}{1-\tau^2}\mathbf{d}\tau.
\]
To transform to the physical coordinates, it is observed that
\[
\tau = \frac{t}{r}, \qquad \rho = \frac{r}{r^2-t^2},
\]
so that
\[
\mathbf{d}\tau =\frac{1}{r}\mathbf{d}t -\frac{t}{r^2}\mathbf{d}r,
\qquad \mathbf{d}r = -\frac{r^2+t^2}{(r^2-t^2)^2}\mathbf{d}r + \frac{2rt}{(r^2-t^2)^2}\mathbf{d}t.
\]
From the above, it follows that
\[
\tilde\bmbeta = \frac{2}{r^2-t^2}\big( t\mathbf{d}t - r\mathbf{d} r  \big).
\]
In particular, one has that
\[
\tilde\bmbeta\big|_{t=0} = -\frac{2}{r}\mathbf{d}r =
\frac{1}{\Omega}\mathbf{d}\Omega, \qquad \Omega\equiv \frac{1}{r^2}.
\]
The conformal factor $\Omega$ realises the \emph{point compactification} of infinity in the Euclidean space. A further
computation shows that
\[
\bmpartial_\tau = \frac{\rho(1+\tau^2)}{\Theta^2}\bmpartial_t + \frac{2\tau\rho}{\Theta^2}\bmpartial_r.
\]
In particular, one has that
\[
\tilde{\bmeta}(\bmpartial_\tau,\bmpartial_\tau) =\frac{1}{\Theta^2}\big(
\rho^2(1+\tau^2)^2 + 4\tau^2\rho^2  \big)\rightarrow \infty, \qquad
\mbox{as} \quad \tau\rightarrow\pm 1. 
\]
The associated integral curves are 
\[
y^\mu(\tau) = \left( \frac{\tau}{\rho_\star(1-\tau^2)},
  \frac{1}{\rho_\star(1-\tau^2)} , \theta_\star, \varphi_\star \right).
\]
As
\[
r(\tau)^2-t(\tau)^2 = \frac{1}{\rho_\star(1-\tau^2)},
\]
one has that these curves do not generate (standard) hyperboloids.



\end{document}